\documentclass[]{eptcs}
\providecommand{\event}{EXPRESS/SOS 2015}

\usepackage[ngerman,english]{babel}
\usepackage[utf8x]{inputenc}
\usepackage{amsmath,amsfonts,amssymb,amsthm}
\usepackage{stmaryrd}
\usepackage{extarrows}
\usepackage{multirow}
\usepackage{xspace}
\usepackage{paralist}
\usepackage{graphicx}
\usepackage{color}
\usepackage{tikz}
\usetikzlibrary{shapes.geometric}
\usetikzlibrary{calc}

\usepackage{count1to}
\usepackage{semantic}
\usepackage{url}
\usepackage[colorinlistoftodos]{todonotes}
\usepackage{cite}
\usepackage{bigdelim}
\usepackage{etoolbox}
\usepackage{hyperref}
\usepackage{cleveref}

\usepackage{macros}

\title{Encoding CSP into CCS (Extended Version)
	\thanks{Supported by the DFG via project ``Synchronous and Asynchronous Interaction in Distributed Systems".}}

\author{Meike Hatzel 
	\institute{TU Berlin}
	\and Christoph Wagner
	\institute{TU Berlin}
	\and Kirstin Peters\thanks{Supported by the German Federal and State Governments via the Excellence Initiative (Institutional Strategy).}
	\institute{TU Dresden}
	\and Uwe Nestmann
	\institute{TU Berlin}
}
\def\titlerunning{Encoding CSP into CCS}
\def\authorrunning{M.\ Hatzel, C.\ Wagner, K.\ Peters, U.\ Nestmann}

\begin{document}

\maketitle

\begin{abstract}
	We study encodings from CSP into asynchronous CCS with name passing and matching, so in fact, the asynchronous $\pi$-calculus. By doing so, we discuss two different ways to map the multi-way synchronisation mechanism of CSP into the two-way synchronisation mechanism of CCS. Both encodings satisfy the criteria of Gorla except for compositionality, as both use an additional top-level context. Following the work of Parrow and Sjödin, the first encoding uses a central coordinator and establishes a variant of weak bisimilarity between source terms and their translations. The second encoding is decentralised, and thus more efficient, but ensures only a form of coupled similarity between source terms and their translations.
\end{abstract}

\section{Introduction}

In the context of a scientific meeting on Expressiveness in Concurrency and Structural Operational Semantics (SOS), likely very little needs to be said about the process algebras (or process calculi) CSP and CCS. Too many papers have been written since their advent in the 70's to be mentioned in our own paper; it is instructive, though, and recommended to appreciate Jos Baeten's historical overview \cite{Baeten:2005:BHP:1085667.1085669}, which also places CSP and CCS in the context of other process algebras like ACP and the many extensions by probabilities, time, mobility, etc. Here, we just select references that help to understand our motivation.

\paragraph{Differences.}
From the beginning, although CSP \cite{hoare:78csp} and CCS \cite{CCS} were intended to capture, describe and analyse reactive and interactive concurrent systems, they were designed following rather different philosophies. Tony Hoare described this nicely in his position paper \cite{Hoare2006209} as follows: ``A primary goal in the original design of CCS was to discover and codify a minimal set of basic primitive agents and operators \dots and a wide range of useful operators which have been studied subsequently are all definable in terms of CCS primitives." and ``CSP was more interested in this broader range of useful operators, independent of which of them might be selected as primitive." So, at their heart, the two calculi use two different synchronisation mechanisms, one (CCS) using binary, \ie two-way, handshake via matching actions and co-actions, the other (CSP) using multiway synchronisation governed by explicit synchronisation sets that are typically attached to parallel composition. Another difference is the focus on Structural Operational Semantics in CCS, and the definition of behavioural equivalences on top of this, while CSP emphasised a trace-based denotational model, enhanced with failures, and the question on how to design models such that they satisfy a given set of laws of equivalence.

\paragraph{Comparisons.}
From the early days, researchers were interested in more or less formal comparisons between CSP and CCS. This was carried out by both Hoare \cite{Hoare2006209} and Milner \cite{DBLP:conf/ifip/Milner86} themselves, where they concentrate on the differences in the underlying design principles. But also other researchers joined the game, but with different analysis tools and comparison criteria. 

For example, Brookes \cite{DBLP:conf/icalp/Brookes83} contributed a deep study on the relation between the underlying abstract models, synchronisation trees for CCS and the failures model of CSP. Quite differently, Lanese and Montanari \cite{Lanese200655} used the respective power to transform graphs as a measure for the expressiveness of the two calculi. 

Yet completely differently, Parrow and Sj{\"o}din \cite{sjodin:phd,parrowCoupled92} tried to find an algorithm to implement|best in a fully distributed fashion|the multiway synchronisation operator of CSP (and its variant LOTOS \cite{DBLP:conf/pstv/Brinksma85}) using the supposedly simpler two-way synchronisation of CCS. They came up with two candidates|a reasonably simple centralised synchroniser, and a considerably less simple distributed synchroniser\footnote{Recently \cite{7092761}, a slight variant of the protocol behind this algorithm was used to implement a distributed compiler for a substantial subset of LOTOS that yields reasonably efficient C code.}|and proved that the two are not weakly bisimilar, but rather coupled similar, which is only slightly weaker. Coupled simulation is a notion that Parrow and Sj\"odin invented for just this purpose, but it has proved afterwards to be often just the right tool when analysing the correctness of distribution- and divergence-sensitive encodings that involve partial commitments (whose only effect is to gradually perform internal choices) \cite{nestmannPierce00}.

The probably most recent comparison between CSP and CCS was provided by van Glabbeek~\cite{DBLP:journals/corr/abs-1208-2750}. As an example for his general framework to analyse the relative expressive power of calculi, he studied the existence of syntactical translations from CSP into CCS, for which a common semantical domain is provided via labeled transition systems (LTS) derived from respective sets of SOS rules. The comparison is here carried out by checking whether a CSP term and its translation into CCS are distinguishable with respect to a number of equivalences defined on top of the LTS. The concrete results are: (1)~there is a translation that is correct up to trace equivalence (and contains deadlocks), and (2)~there is no translation that is correct up to weak bisimilarity equivalence that also takes divergence into account.

\paragraph{Contribution.}
Given van Glabbeks negative result, and given Parrow and Sj\"odins algorithm, we set out to check whether we can define a syntactical encoding from CSP into CCS|using Parrow and Sj\"odins ideas|that is correct up to coupled simulation. We almost managed. In this paper, we report on our current results along these lines: 
(1)~Our encoding target is an asynchronous variant of CCS, but enhanced with name-passing and matching, so it is in fact the asynchronous $\pi$-calculus; we kept mentioning CCS in the title of this paper, as it clearly emphasises the origin and motivation of this work. But, we could \emph{not} do without name-passing.
(2)~We exhibit one encoding that is not distributability-preserving (so, it represents a centralised solution), but is correct up to weak bisimilarity and does not introduce divergence. This does not contradict van Glabbeek's results, but suggests the observation that van Glabbeek's framework implies some form of distributability-preservation.
(3)~We exhibit another encoding that \emph{is} distributability-preserving and divergence-reflecting, but is only correct up to coupled similarity.

\paragraph{Overview.}
This paper is an extended version---providing the missing proofs and some additional in\-for\-ma\-tions---of \cite{hatzel15}.
We introduce the considered variants of CSP and CCS in \S~\ref{sec:techPrel}. There we also introduce the criteria---that are (variants of) the criteria in \cite{gorla10} and \cite{petersNestmannGoltz13}---modulo which we prove the quality of the considered encodings. In \S~\ref{sec:innerPart} we introduce the inner layer of our two encodings. It provides the main machinery to encode synchronisations of CSP. We complete this encoding with an outer layer that is either a central (\S~\ref{sec:central}) or a de-central coordinator (\S~\ref{sec:decentral}). In \S~\ref{sec:conclusion} we discuss the two encodings.

\section{Technical Preliminaries}
\label{sec:techPrel}

Let $\Names$ be the countably-infinite set of names, let $\tau \not\in \Names$, and let $\CoNames$ the set of co-names, \ie $\CoNames = \{ \overline{a} \mid a \in \Names \}$.
A process calculus $ \left( \proc, \step \right) $ consists of a set of processes $ \proc $ (syntax) and a relation $ {\step} \subseteq \proc^2 $ (semantics).
$ \tau $ denotes an internal, \ie unobservable, action.
We use $ a, b, x, \ldots $ to range over names and $ P, Q, \ldots $ to range over processes.
We use $\alpha, \beta \ldots$ to range over $\Names \cup \{\tau\}$.
$\tilde{a}$ denotes a sequence of names.
Let $ \fn{P} $ and $ \bn{P} $  denote the sets of free names and bound names occurring in $P$, respectively.
Their definitions are completely standard.
We use $\sigma, \sigma', \sigma_{1}, \ldots$ to range over substitutions.
A substitution is a mapping $[\Subst{x_{1}}{y_{1}}, \ldots ,\Subst{x_{n}}{y_{n}}]$ from names to names.
The application of a substitution on a term $P[\Subst{x_{1}}{y_{1}}, \ldots ,\Subst{x_{n}}{y_{n}}]$ is defined as the result of simultaneously replacing all free occurrences of $y_{i}$ by $x_{i}$ for $i \in \Set{1, \ldots, n}$.
For all names in $\Names \setminus \Set{y_{1}, \ldots, y_{n}}$ the substitution behaves as the identity mapping.
We naturally extend substitutions to co-names, \ie $\forall \overline{a} \in \Names.\; \sigma(\overline{a}) = \overline{\sigma(a)}$ for all substitutions $\sigma$.
The relation $\step$ as defined in the semantics below defines the reduction steps processes can perform. We write $ P \step P' $ if $ (P, P') \in \step $ and call $ P' $ a \emph{derivative} of $ P $.
Let $ \steps $ denote the reflexive and transitive closure of $ \step $.
$ P $ is \emph{divergent} if it has an infinite sequence of steps $ P \step^{\omega} $.
We use \emph{barbs} or \emph{observables} to distinguish between processes with different behaviours. We write $ P\HasBarb{\alpha} $ if $ P $ has a barb $ \alpha $, where the predicate $ \cdot\HasBarb{\cdot} $ can be defined different for each calculus. Moreover $ P $ reaches a barb $ \alpha $, if $ P $ a reaches a process with this barb, \ie $ P\ReachBarb{\alpha} \deff \exists P'\logdot P \steps P' \wedge P'\HasBarb{\alpha} $.

We use a variant of CSP \cite{hoare:78csp}, where prefixes only occur behind external choice.
\begin{definition}\label{CSPSyntax}
The processes $\ProcCSP$ are given by
$$P \;\mathop{::=}\; \CSPPar{P}{P}{A} \sep \Div \sep \Stop \sep \inchoice{P}{P} \sep P/b \sep f(P) \sep X \sep \mu X \cdot P \sep \textstyle{\ExSum_{i \in \IM}} a \rightarrow P.$$
	where $X \in \mathcal{X}$ is a process variable, $ A \subseteq \names $, and $ \IM $ is an index set.
\end{definition}

$\CSPPar{P}{Q}{A}$ is the parallel composition of $ P $ and $ Q $, where $ P $ and $ Q $ can proceed independently except for actions $ a \in A $, on which they have to synchronise.
The process $\Div$ describes \emph{divergence}.
$\Stop$ denotes \emph{inaction}.
The \emph{internal choice} operator $\inchoice{P}{Q}$ reduces to either $ P $ or $ Q $ within a single internal step.
\emph{Concealment} $P/b$ hides an action $ b $ and masks it as $ \tau $.
\emph{Renaming} $f(P)$ for some $f:\Names \to \Names$ with $f(\tau)=\tau$ behaves as $ P $, where $ a $ is replaced by $ f(a) $ for all $ a \in \Names $.
The \emph{recursion} $\mu X \cdot P$ describes a process behaving like $P$ with every occurrence of $X$ being replaced by $\mu X \cdot P$.
Finally the \emph{external choice} $\ExSum_{i \in \IM} a_i \rightarrow P_i$ represents a choice between the different \emph{action prefixes} $a_i \rightarrow \cdot $ followed by the corresponding continuation $P_i$. The process can perform any $ a_i $ and then behave like $P_i$.
 
As usual we use $\exchoice{M}{N}$ to denote binary external choice.
The CSP semantics is given by the rules, where we introduce labelled steps $\Trans{\alpha}$ first and then use them to define $\step$:
\vspace{-0.75em}
\begin{align*}
		\begin{array}{|c|}
			\hline
			\ninfer{\runa{Con$_0$}}{E \Trans{b} E'}{E /b \Trans{\tau} E'/b} \hspace*{2.5em} \ninfer{\runa{Con$_1$}}{E \Trans{\alpha} E' \quad (\alpha \not = b)}{E /b \Trans{\alpha} E'/b}  \hspace*{2.5em} \ninfer{\runa{Ren}}{E \Trans{\alpha} E'}{f(E) \Trans{f(\alpha)} f(E')} \hspace*{2.5em} \ninfer{\runa{Ece}}{M_j \Trans{a} M_j'}{\ExSum_{i \in  \IM}M_i \Trans{a} M_j'}\\
			\\
			\ninfer{\runa{Act}}{}{(a \rightarrow E) \Trans{a} E} \hspace*{2.5em} \ninfer{\runa{Rec}}{}{\mu X \cdot E \Trans{\tau} E[\mu X \cdot E/X]}\\
			\\
			\ninfer{\runa{Par$_0$}}{E \Trans{\alpha} E' \quad (\alpha \not \in A)}{\CSPPar{E}{F}{A} \Trans{\alpha} \CSPPar{E'}{F}{A}} \hspace*{1.5em} \ninfer{\runa{Par$_1$}}{F \Trans{\alpha} F' \quad (\alpha \not \in A)}{\CSPPar{E}{F}{A} \Trans{\alpha} \CSPPar{E}{F'}{A}} \hspace*{1.5em} \ninfer{\runa{Par$_2$}}{F \Trans{a} F' \quad E \Trans{a} E' \quad (a \in A)}{\CSPPar{E}{F}{A} \Trans{a} \CSPPar{E'}{F'}{A}} \hspace*{1.5em} \ninfer{\runa{Red}}{P \Trans{\alpha} P'}{P \step P'}\\
			\\
			\ninfer{\runa{Div}}{}{\Div \Trans{\tau} \Div} \hspace*{2.5em} \ninfer{\runa{Inc$_0$}}{}{\inchoice{E}{F} \Trans{\tau} E} \hspace*{2.5em} \ninfer{\runa{Inc$_1$}}{}{\inchoice{E}{F} \Trans{\tau} F} 
			\\
			\hline
		\end{array}
\end{align*}
\vspace{-1.2em}\\
A barb of CSP is the possibility of a term, to perform an action, \ie $ P\HasBarb{a} \deff \exists P'\logdot P \Trans{a} P' $.
Following the Definition of distributability in \cite{petersNestmannGoltz13} a CSP term $ P $ is distributable into $ P_1, \ldots, P_n $ if $ P_1, \ldots, P_n $ are subterms of $ P $ such that each action prefixes in $ P $ occurs exactly in one of the $ P_1, \ldots, P_n $, where different but syntactically equivalent action prefixes are distinguished and unguarded occurrences of $ \mu X \cdot P' $ may result in several copies of $ P' $ within the $ P_1, \ldots, P_n $.

As target calculus we use an asynchronous variant of CCS\cite{CCS} with name-passing and matching.

\begin{definition}\label{def:ccs_syntax}
  The processes $\ProcCCS$ are given by
  \begin{align*}
    P & \;\mathop{::=}\; \Par{P}{P} \sep \Res{\tilde{c}}{P} \sep \RepInput{\ch}{\tilde{x}}{P} \sep \Input{\ch}{\tilde{x}}{P} \sep \Out{\ch}{\tilde{x}} \sep \Match{\ch}{z}P \sep \Null
  \end{align*}
\end{definition}

$\Par{P}{Q}$ is the parallel composition of $P$ and $Q$, where $P$ and $Q$ can either proceed independently or synchronise on channels with matching names.
$\Res{\tilde{a}}{P}$ restricts actions $\tilde{a}$ on $P$, forcing all sub-terms of $P$ to synchronise on these actions.
$ \Input{a}{\tilde{x}}{P} $ denotes input on channel $a$.
$ \Out{a}{\tilde{x}} $ is output on channel $a$.
Note that because there is no continuation we interpret this calculus as asynchronous.
We use $ \RepInput{\ch}{\tilde{x}}{P} $ to denote \emph{replicated input} on channel $c$ with the continuation $P$.
$ \Match{x}{y}P  $ is the matching operator, if $ x = y $ then $P$ is enabled.
$\Null$ denotes inaction.

The CCS semantics is given by following transition rules:
\vspace{-0.6em}
\begin{align*}
		\begin{array}{|c|}
			\hline
			\ninfer{\runa{Par}}{P \step P'}{\Par{P}{Q} \step \Par{P'}{Q}} \hspace*{2.5em} \ninfer{\runa{Res}}{P \step P'}{\Res{\tilde{c}}{P} \step \Res{\tilde{c}}{P'}} \hspace*{2.5em} \ninfer{\runa{Cong}}{P \equiv P' \quad P' \step Q' \quad Q' \equiv Q}{P \step Q}\\
			\\
			\ninfer{\runa{Rep}}{}{\Par{\RepInput{\ch}{\tilde{x}}P}{\Out{\ch}{\tilde{y}}} \step \Par{\RepInput{\ch}{\tilde{x}}P}{P[\tilde{y}/\tilde{x}]}} \hspace*{2.5em} \ninfer{\runa{Com}}{}{\Par{\Out{\ch}{\tilde{y}}}{\Input{\ch}{\tilde{x}}Q} \step \Par{P}{Q[\tilde{y}/\tilde{x}]}}
			\\
			\hline
		\end{array}
\end{align*}
\vspace{-1em}\\
where $ \equiv $ denotes structural congruence given by the rules: $ \Par{P}{0} \equiv P$,
$ \Par{P}{Q} \equiv \Par{Q}{P} $, $ \Par{P}{\left( \Par{Q}{R} \right)} \equiv \Par{\left(\Par{P}{Q} \right)}{R} $, $\Res{\tilde{a}}{\Null} \equiv \Null $, $\Par{P}{\Res{\tilde{a}}{Q}} \equiv \Res{\tilde{a}}{\Par{P}{Q}} $ if $\bn{\tilde{a}} \notin \fn{P}$, $\Match{x}{x}P \equiv P $ and $ \Match{x}{y}P \equiv \Null $ if $ x \neq y $.
As discussed in \cite{petersNestmannGoltz13} a CCS term $ P $ is distributable into $ P_1, \ldots, P_n $ if $ P \equiv \Res{\tilde{x}}{P_1 \mid \ldots \mid P_n} $.

\subsection{Simulation Relations}

The semantics of a process is usually considered modulo some behavioural equivalence on processes.
For many calculi \emph{the} standard reference equivalence is some form of bisimulation.
Since in the context of encodings, \ie translation between different languages that can differ in their interpretation of what is considered a barb, reduction steps are easier, we use a variant of weak reduction bisimulation.
With Gorla \cite{gorla10}, we add a \emph{success} operator $ \success $ to the syntax of both CSP and CCS. Since $ \success $ cannot be further reduced, the semantics is left unchanged in both cases. The test for the reachability of success is standard in both languages, \ie $ P\hasSuccess \deff \exists P'\logdot P \equiv \success \mid P' $.
To obtain a non-trivial equivalence, we require that the bisimulation respects success and the reachability of barbs.
Therefore we use the standard definition of barbs in CSP, \ie action-prefixes, for CSP-barbs.
Our encoding function will translate all source terms into closed terms, thus the standard definition of CCS barbs would not provide any information.
Instead we use a notion of translated barb ($ \cdot\ReachBarb{\EncCI{\cdot}\cdot} $) that reflects how the encoding function translates source term barbs. Its definition is given in Section~\ref{sec:innerPart}.

\begin{definition}[Bisimulation]
	A relation $ \mathcal{R} \; \subseteq \proc^2 $ is a \emph{success sensitive, (translated) barb respecting, weak, reduction bisimulation} if, whenever $ \left( P, Q \right) \in \mathcal{R} $, then:
	\begin{compactitem}
		\item $ P \step P' $ implies $ \exists Q'\logdot Q \steps Q' \wedge \left( P', Q' \right) \in \mathcal{R} $
		\item $ Q \step Q' $ implies $ \exists P'\logdot P \steps P' \wedge \left( P', Q' \right) \in \mathcal{R} $
		\item $ P\reachSuccess $ iff $ Q\reachSuccess $
		\item $ P $ and $ Q $ reach the same (translated) barbs, where we use $ \cdot\ReachBarb{a} $ for CSP and $ \cdot\ReachBarb{\EncCI{\cdot}a} $ for CCS
	\end{compactitem}
	Two terms $ P, Q \in \proc $ are \emph{bisimilar}, denoted as $ P \barbBisim Q $, if there exists a bisimulation that relates $ P $ and $ Q $.
\end{definition}

\noindent
We use the symbol $ \barbBisim $ to denote either bisimilarity on our target language CCS or on the disjoint union of CSP and CCS that allows us to describe the relationship between source terms and their translations. In the same way we define a corresponding variant of coupled similarity.

\begin{definition}[Coupled Simulation]
	A relation $ \mathcal{R} \; \subseteq \proc^2 $ is a \emph{success sensitive, (translated) barb respecting, weak, reduction coupled simulation} if, whenever $ \left( P, Q \right) \in \mathcal{R} $, then:
	\begin{compactitem}
		\item $ P \step P' $ implies $ \exists Q'\logdot Q \steps Q' \wedge \left( P', Q' \right) \in \mathcal{R} $ and $ \exists Q''\logdot Q \steps Q'' \wedge \left( Q'', P' \right) \in \mathcal{R} $
		\item $ P\reachSuccess $ iff $ Q\reachSuccess $
		\item $ P $ and $ Q $ reach the same (translated) barbs, where we use $ \cdot\ReachBarb{a} $ for CSP and $ \cdot\ReachBarb{\EncCI{\cdot}a} $ for CCS
	\end{compactitem}
	Two terms $ P, Q \in \proc $ are \emph{coupled similar}, denoted as $ P \barbCS Q $, if there exists a coupled simulation that relates $ P $ and $ Q $ in both directions.
\end{definition}

\subsection{Encodings and Quality Criteria}

We consider two different translations from (the above defined variant of) CSP into (the above defined variant of) CCS with name passing and matching. We denote the variant of CSP as \emph{source} and the variant of CCS as \emph{target} language and, accordingly, their terms as \emph{source terms} $ \procS $ and \emph{target terms} $ \procT $. Encodings often translate single source term steps into a sequence or pomset of target term steps. We call such a sequence or pomset a \emph{simulation} of the corresponding source term step.
Moreover we assume for each encoding the existence of a so-called renaming policy $ \varphi $, \ie a mapping of names from the source into vectors of target term names.

To analyse the quality of encodings and to rule out trivial or meaningless encodings, they are augmented with a set of quality criteria. In order to provide a general framework, Gorla in \cite{gorla10} suggests five criteria well suited for language comparison:
\begin{compactenum}[(1)]
	\item \emph{Compositionality}: The translation of an operator $ \mathrm{op} $ is the same for all occurrences of that operator in a term, \ie it can be captured by a context $ \context_{\mathrm{op}} $ such that $ \Enc{\mathrm{op}\left( x_1, \ldots, x_n, S_1, \ldots, S_m \right)} = \Context{N}{\mathbf{op}}{x_1, \ldots, x_n, \Enc{S_1}, \ldots, \Enc{S_m}} $ for $ \fn{S_1} \cup \ldots \cup \fn{S_m} = N $.
	\item \emph{Name Invariance}: The encoding does not depend on particular names, \ie for every $ S $ and $ \sigma $, it holds that $ \Enc{\sigma\left( S \right)} \equiv \sigma'\left( \Enc{S} \right) $ if $ \sigma $ is injective and $ \Enc{\sigma\left( S \right)} \asymp \sigma'\left( \Enc{S} \right) $ otherwise, where $ \sigma' $ is such that $ \Renam{\sigma\left( n \right)} = \sigma'\left( \Renam{n} \right) $ for every $ n \in \names $.
	\item \emph{Operational Correspondence}: Every computation of a source term can be simulated by its translation, \ie $ S \stepsS S' $ implies $ \Enc{S} \stepsT \asymp \Enc{S'} $ (completeness), and every computation of a target term corresponds to some computation of the corresponding source term (soundness, compare to Section~\ref{sec:decentral}).
	\item \emph{Divergence Reflection}: The encoding does not introduce divergence, \ie $ \Enc{S} \stepT^{\omega} $ always implies $ S \stepS^{\omega} $.
	\item \emph{Success Sensitiveness}: A source term and its encoding answer the tests for success in exactly the same way, \ie $ S \reachSuccess $ iff $ \Enc{S} \reachSuccess $.
\end{compactenum}
Operational correspondence and name invariance assume a behavioural equivalence $ \asymp $ on the target language (that we instantiate with $ \barbBisim $). Its purpose is to describe the abstract behaviour of a target process, where abstract refers to the behaviour of the source term. By \cite{gorla10} the equivalence $ \asymp $ is often defined in the form of a barbed equivalence (as described e.g. in \cite{milner.sangiorgi:barbed-bisimulation}) or can be derived directly from the reduction semantics and is often a congruence, at least with respect to parallel composition. $ \barbBisim $ is such a relation.

Our encodings will satisfy all of these criteria except for compositionality, because both encodings consists of two layers.
\cite{petersNestmannGoltz13} shows that the above criteria do not ensure that an encoding preserves distribution and proposes a criterion for the preservation of distributability.

\begin{definition}[Preservation of Distributability]
	\label{def:distributabilityPreservation}

	An encoding $ \enc $ \emph{preserves distributability} if for every $ S $ and for all terms $ S_1, \ldots, S_n $ that are distributable within $ S $ there are some $ T_1, \ldots, T_n $ that are distributable within $ \Enc{S} $ such that $ T_i \asymp \Enc{S_i} $ for all $ 1 \leq i \leq n $.
\end{definition}

\noindent
Here, because of the choice of the source and the target language, an encoding preserves distributability if for each sequence of distributable source term steps their simulations are pairwise distributable. In both languages two alternative steps of a term are in \emph{conflict} with each other if they reduce the same action-prefix---for CSP---or reduce either the same (replicated) input using two outputs that transmit different values, or reduce the same output using two (replicated) inputs with different continuations. Two alternative steps that are not in conflict are \emph{distributable}.

\section{Translating the CSP Synchronisation Mechanism}
\label{sec:innerPart}

CSP and CCS---or the $ \pi $-calculus---differ fundamentally in their communication and synchronisation mechanisms.
In CSP there is only a single kind of action $ c \rightarrow \cdot $, where $ c $ is a (channel) name. Synchronisation is implemented by the parallel operator $ \CSPPar{\cdot}{\cdot}{A} $ that in CSP is augmented with a set of names $ A $ containing the names that need to be synchronised at this point. By nesting parallel operators arbitrary many actions on the same name can be synchronised.
In CCS there are two different kinds of actions: inputs $ \In{c}{} $ and outputs $ \Out{c}{} $. Again synchronisation is implemented by the parallel operator, but in CCS only a single input and a single matching output can ever be synchronised.

To encode the CSP communication and synchronisation mechanisms in CCS with name passing we make use of a technique already used in \cite{petersNestmann12, peters12} to translate between different variants of the $ \pi $-calculus. CSP actions are translated into action announcements augmented with a lock indicating, whether the respective action was already used in the simulation of a step. The other operators of CSP are then translated into handlers for these announcements and locks.
The translation of sum combines several actions under the same lock and thus ensures that only one term of the sum can ever be used.
The translation of the parallel operator combines announcements of actions that need to be synchronised into a single announcement under a fresh lock, whose value is determined by the combination of the respective underlying locks at its left and right side. Announcements of actions that do not need to be synchronised are simply forwarded.
A second layer---containing either a centralised or a de-centralised coordinator---then triggers and coordinates the simulation of source term steps.

Action announcements are of the form $ \Out{\act}{\ch, \req, \lock, \simu} $: $ \ch $ is the translation of the source term action. $ \req $ is used to trigger the computation of the Boolean value of $ \lock $. The lock $ \lock $ evaluates to $ \top $ as long as the respective translated action was not successfully used in the simulation of a step. $ \simu $ is used to guard the encoded continuation of the respective source term action. In the case of a successful simulation attempt involving this announcement, an output $ \Out{\simu}{\top} $ allows to unguard the encoded source term continuation and ensures that all following evaluations of $ \lock $ return $ \bot $. The message $ \Out{\simu}{\bot} $ indicates an aborted simulation attempt and allows to restore $ \lock $ for later simulation attempts. Once a lock becomes $ \bot $, all request for its computation return $ \bot $.

\subsection{Abbreviations}

We introduce some abbreviations to simplify the presentation of the encodings. We use
\begin{align*}
	\MatchIn{x}{A}P \deff \textstyle{\prod_{a \in A}} \Match{x}{a}P
\end{align*}
to test, whether an action belongs to the set of synchronised actions in the encoding of the parallel operator.
As already done in \cite{nestmann96, nestmannPierce00} we use Boolean valued locks to ensure that every translation of an action is only used once to simulate a step.
\emph{Boolean locks} are channels on which only the Boolean values $ \top $ (true) or $ \bot $ (false) are transmitted. An unguarded output over a Boolean lock with value $ \top $ is called a positive instantiation of the respective lock, whereas an unguarded output sending $ \bot $ is denoted as negative instantiation. At the receiving end of such a channel, the Boolean value can be used to make a binary decision, which is done here within an \emph{$ \ITE{\cdot}{\cdot}{\cdot} $-construct}.
This construct and accordingly instantiations of locks are implemented as in \cite{nestmann96, nestmannPierce00} using restriction and the order of transmitted values.
\begin{align*}
	\Out{l}{\top} \deff \Input{l}{t, f}{\Out{t}{}} \quad & \quad\quad \Out{l}{\bot} \deff \Input{l}{t, f}{\Out{f}{}}\\
	\Input{l}{\boolV}{\ITE{\boolV}{P}{Q}} & \deff \Res{t, f}{\Out{l}{t, f} \mid \Input{t}{}{P} \mid \Input{f}{}{Q}}
\end{align*}
We observe that the Boolean values $ \top $ and $ \bot $ are realised by a pair of links without parameters. Both cases of the $ \ITE{\cdot}{\cdot}{\cdot} $-construct operate as guard for its subterms $ P $ and $ Q $. The renaming policy $ \renam $ reserves the names $ t $ and $ f $ to implement the Boolean values $ \true $ and $ \false $.

\subsection{The Algorithm}

The encoding functions introduce some fresh names, that are reserved for special purposes. In Table~\ref{tab:resNam} we list the reserved names $ \mathcal{R} $ and provide a hint on their purpose.
\begin{table}[t]
	\begin{tabular}{|c|c|}
		\hline
		reserved names & purpose\\
		\hline
		$ \act $, $ \act' $ & announce the ability to perform an action\\
		$ \ch $, $ \lch $, $ \rch $, $ z $ & (translated) source term channel, channel from the left/right of a parallel operator\\
		$ \lock $, $ \lLock $, $ \rLock $ & lock, lock from the left/right of a parallel operator\\
		$ \lock' $ & re-instantiate a positive sum lock\\
		$ \req $, $ \lreq $, $ \rreq $ & request the computation of the value of a lock\\
		$ \simu $, $ \simu_i $, $ \lsimu $, $ \rsimu $ & simulate a source term step and unguard the corresponding continuations\\
		$ \nextSyn $ & order left announcements for the same channel that need to be synchronised\\
		$ \syn $, $ \syn' $ & distribute right announcements that need to be synchronised\\
		$ \boolV $ & Boolean value ($ \bot $ or $ \top $)\\
		$ \tau $ & fresh name used to announce $ \tau $-steps that result from concealment\\
		$ \once $ & used by the centralised encoding to avoid overlapping simulation attempts\\
		$ \much $ & fresh names used to encode internal choice\\
		$ \rep $ & fresh names used to encode divergence\\
		$ t, f $ & used to encode Boolean values\\
		\hline
	\end{tabular}
	\caption{Reserved Names.}
	\label{tab:resNam}
\end{table}
Moreover we reserve the names $ \Set{ x_i \mid i \in \nat } $ and assume an injective mapping $ \renam': \mathcal{X} \to \Set{ x_i \mid i \in \nat } $ that maps process variables of CSP to distinct names.
The renaming policy $ \renam $ for our encodings is then a function that reserves the names in $ \mathcal{R} \cup \Set{ x_i \mid i \in \nat } $ and translates every source term name into three target term names. More precisely, choose $ \varphi : \names \to \names^3 $ such that:
\begin{compactenum}
	\item No name is mapped onto a reserved name, \ie $ \Renam{n} \cap \left( \mathcal{R} \cup \Set{ x_i \mid i \in \nat } \right) = \emptyset $ for all $ n \in \names $.
	\item No two different names are mapped to overlapping sets of names, \ie $ \Renam{n} \cap \Renam{m} = \emptyset $ for all $ n, m \in \names $ with $ n \neq m $.
\end{compactenum}
We naturally extend the renaming policy to sets of names, \ie $ \Renam{X} \deff \Set{ \Renam{x} \mid x \in X } $ if $ X \subseteq \names $.
Let $ \Proj{\left( x_1, \ldots, x_n \right)}{i} \deff x_i $ denote the projection of a $ n $-tuple to its $ i $th element, if $ 1 \leq i \leq n $. Moreover $ \Proj{X}{i} \deff \Set{ \Proj{x}{i} \mid x \in X } $ for a set of $ n $-tuples $ X $ and $ 1 \leq i \leq n $.

\begin{figure}[htp]
	\begin{align*}
		\EncCI{\CSPPar{P}{Q}{A}} \deff & \res{\act', \Proj{\Renam{A}}{2}, \Proj{\Renam{A}}{3}}\Big(\!\\
			& \quad \Res{\act}{\EncCI{P} \mid \RepInput{\act}{\ch, \tilde{x}}{\left( \MatchIn{\ch}{\Proj{\Renam{A}}{1}} \Out{\Proj{\Renam{\ch}}{2}}{\tilde{x}} \mid \MatchOut{\ch}{\Proj{\Renam{A}}{1}} \Out{\act'}{\ch, \tilde{x}} \right)}}\\
			& \quad \Res{\act}{\EncCI{Q} \mid \RepInput{\act}{\ch, \tilde{x}}{\left( \MatchIn{\ch}{\Proj{\Renam{A}}{1}} \Out{\Proj{\Renam{\ch}}{3}}{\tilde{x}} \mid \MatchOut{\ch}{\Proj{\Renam{A}}{1}} \Out{\act'}{\ch, \tilde{x}} \right)}}\\
			& \quad \mid \textstyle{\prod_{\ch \in A}} \Synch{\ch} \mid \RepInput{\act'}{\tilde{x}}{\Out{\act}{\tilde{x}}}
		\!\Big)
		\\
		\Synch{\ch} \deff & \res{\nextSyn}\Big(\!
			\Out{\nextSyn}{\Proj{\Renam{\ch}}{3}}\\
			& \quad \mid \RepIn{\nextSyn}{\syn}\Big(\In{\Proj{\Renam{\ch}}{2}}{\lreq, \lLock, \lsimu}.\Big( \res{\syn'}\Big(\\
				& \quad\quad \RepInput{\syn}{\rreq, \rLock, \rsimu}{\left( \Res{\req, \lock, \simu}{\Out{\act}{\Proj{\Renam{c}}{1}, \req, \lock, \simu} \mid \Sim} \mid \Out{\syn'}{\rreq, \rLock, \rsimu} \right)}\\
				& \quad\quad \mid \Res{\syn}{\Out{\nextSyn}{\syn} \mid \RepInput{\syn'}{\tilde{x}}{\Out{\syn}{\tilde{x}}}}
			\Big)\!\Big)\!\Big)
		\!\Big)
		\\
		\Sim \deff & \res{\lock'}\Big( \Out{\lock'}{} \mid \RepIn{\lock'}{}.\Big( \In{\req}{}.\Big(\Out{\lreq}{} \mid \In{\lLock}{\boolV}.\Big( \IF{\boolV}\;\THEN{}\Big( \Out{\rreq}{} \mid \In{\rLock}{\boolV}.\Big( \IF{\boolV}\\
				& \quad \quad \THEN{\left( \Out{\lock}{\top} \mid \Input{\simu}{\boolV}{\left( \Out{\lsimu}{\boolV} \mid \Out{\rsimu}{\boolV} \mid \ITE{\boolV}{\RepInput{\req}{}{\Out{\lock}{\bot}}}{\Out{\lock'}{}} \right)} \right)}\\
				& \quad \quad \ELSE{\left( \Out{\lock}{\bot} \mid \Out{\lsimu}{\bot} \right) \mid \RepInput{\req}{}{\Out{\lock}{\bot}}} \Big)\!\Big)\\
			& \quad \ELSE{\left( \Out{\lock}{\bot} \mid \RepInput{\req}{}{\Out{\lock}{\bot}} \right)}\!\Big)\!\Big)\!\Big)\!\Big)
		\\
		\EncCI{\textstyle{\ExSum_{i \in \IM}} \ch_i \to P_i} \deff & \res{\req, \lock, \simu_1, \ldots, \simu_n}\Big( \Input{\req}{}{\Out{\lock}{\top}}\\
			& \mid \textstyle{\prod_{i \in \IM}} \left( \Out{\act}{\Proj{\ch_i}{1}, \req, \lock, \simu_i} \mid \RepIn{\simu_i}{\boolV}.\ITE{\boolV}{\left( \EncCI{P_i} \mid \RepInput{\req}{}{\Out{\lock}{\bot}} \right)}{\Input{\req}{}{\Out{\lock}{\top}}} \right)\!\Big)
		\\
		\EncCI{\left( P \right) / z} \deff & \Res{\act'}{\Res{\act, z}{\EncCI{P} \mid \RepInput{\act}{\ch, \tilde{x}}{\left( \Match{\ch}{z}\Out{\act'}{\tau, \tilde{x}} \mid \MissMatch{\ch}{z}\Out{\act'}{\ch, \tilde{x}} \right)}} \mid \RepInput{\act'}{\tilde{x}}{\Out{\act}{\tilde{x}}}}
		\\
		\EncCI{f(P)} \deff & \res{\act'}\big( \res{\act, z}\big( \EncCI{P} \mid \RepIn{\act}{\ch, \tilde{x}}\big( \textstyle{\prod_{z/x \in f}} \Match{\ch}{\Proj{\Renam{x}}{1}}\Out{\act'}{\Proj{\Renam{z}}{1}, \tilde{x}}\\
		& \mid \MatchOut{\ch}{\Dom{f}}\Out{\act'}{\ch, \tilde{x}} \big)\!\big) \mid \RepInput{\act'}{\tilde{x}}{\Out{\act}{\tilde{x}}}\big)
		\\
		\EncCI{\Div} \deff & \Res{\rep}{\Out{\rep}{} \mid \RepInput{\rep}{}{\Out{\rep}{}}}
		\\
		\EncCI{\mu X \cdot P} \deff & \Res{\renam'(X)}{\Out{\renam'(X)}{} \mid \RepInput{\renam'(X)}{}{\EncCI{P}}}
		\\
		\EncCI{X} \deff & \Out{\renam'(X)}{}
		\\
		\EncCI{\inchoice{P}{Q}} \deff & \Res{\much}{\Input{\much}{}{\EncCI{P}} \mid \Input{\much}{}{\EncCI{Q}} \mid \Out{\much}{}}
		\\
		\EncCI{\Stop} \deff & \Null
		\\
		\EncCI{\success} \deff & \success
	\end{align*}
	where $ \notin \Proj{\Renam{A}}{1} $ is short for $ \in \left( \fn{P} \cup \fn{Q} \right) \setminus \Proj{\Renam{A}}{1} $, $ \notin \Dom{f} $ is short for $ \in \fn{P} \setminus \Dom{f} $, and $ \neq z $ is short for $ \in \fn{P} \setminus \Set{ z } $.
	\caption{An encoding from CSP into CCS with value passing (inner part).}
	\label{fig:innerEncoding}
\end{figure}
The inner part of our two encodings is presented in Figure~\ref{fig:innerEncoding}. The most complex case is the translation of the parallel operator $ \EncCI{\CSPPar{P}{Q}{A}} $ that is based on the following four steps:
\begin{description}
	\item[Step 1:] Action announcements for channels $ \ch \notin A $\\
		In the case of actions on channels $ \ch \notin A $---that do not need to be synchronised here---the encoding of the parallel operator acts like a forwarder and transfers action announcements of both its subtrees further up in the parallel tree.
		Two different restrictions of the channel for action announcements $ \act $ from the left side $ \EncCI{P} $ and the right side $ \EncCI{Q} $, allow to trace action announcements back to their origin as it is necessary in the following case.
		In the present case we use $ \act' $ to bridge the action announcement over the restrictions on $ \act $.
	\item[Step 2:] Action announcements for channels $ \ch \in A $\\
		Actions $ \ch \in A $ need to be synchronised, \ie can be performed only if both sides of the parallel operator cooperate on this action. Simulating this kind of synchronisation is the main purpose of the encoding of the parallel operator.
		The renaming policy $ \renam $ translates each source term name into three target term names. The first target term name is used as reference to the original source term name and transferred in announcements. The other two names are used to simulate the synchronisation of the parallel operator in CSP. Announcements from the left are translated to outputs on the respective second name and announcements from the right to the respective third name. Restriction ensures that these outputs can only be computed by the current parallel operator encoding. The translations of the announcements into different outputs for different source term names allows us to treat announcements of different names concurrently using the term $ \Synch{\ch} $, where $ \ch $ is a source term name.
	\item[Step 3:] The term $ \Synch{\ch} $\\
		In $ \Synch{\ch} $ all announcements for the same source term name $ \ch $ from the left are ordered in order to combine each left and each right announcement on the same name. Several such announcements may result from underlying parallel operators, sums with similar summands, and junk left over from already simulated source term steps. For each left announcement a fresh instance of $ \syn $ is generated and restricted. The names $ \syn $ and $ \syn' $ are used to transfer right announcements to the respective next left announcement, where $ \syn' $ is used to bridge over the restriction on $ \syn $. This way each right announcement will eventually be transferred to each left announcement on the same name. Note that this kind of forwarding is not done concurrently but in the source language a term $ \CSPPar{P}{Q}{A} $ also cannot perform two steps on the same name $ \ch \in A $ concurrently. After combining a left and a right announcement on the same source term name a fresh set of auxiliary variables $ \req, \lock, \simu $ is generated and a corresponding announcement is transmitted. The term $ \Sim $ reacts to requests regarding this announcement and is used to simulate a step on the synchronised action.
	\item[Step 4:] The term $ \Sim $\\
		If a request reaches $ \Sim $ it starts questioning the left and the right side. First the left side is requested to compute the current value of the lock of the action. Only if $ \top $ is returned, the right side is requested to compute its lock as well. This avoids deadlocks that would result from blindly requesting the computation of locks in the de-centralised encoding. If the locks of both sides are still valid the fresh lock $ \lock $ returns $ \top $ else $ \bot $ is returned. For each case $ \Sim $ ensures that subsequently requests will obtain an answer by looping with $ \Out{\lock'}{} $ or returning $ \bot $ to all requests, respectively. The messages $ \Out{\lsimu}{\bot} $ and $ \Out{\rsimu}{\bot} $ cause the respective underlying subterms on the left and the right side to do the same, whereas $ \Out{\lsimu}{\top} $ and $ \Out{\rsimu}{\top} $ cause the unguarding of encoded continuations as result of a successful simulation of a source term synchronisation step.
\end{description}

\subsection{Basic Properties and Translated Observables}

The protocol introduced by the encoding function in Figure~\ref{fig:innerEncoding} (and its outer parts introduced later) simulates a single source term step by a sequence of target term steps. Most of these steps are merely pre- and post-processing steps, \ie they do not participate in decisions regarding the simulation of conflicting source term steps but only prepare and complete simulations. Accordingly we distinguish between \emph{auxiliary steps}---that are pre- and post-processing steps---and \emph{simulation steps}---that mark a point of no return by deciding which source term step is simulated. Note that the points of no return and thus the definition of auxiliary and simulation steps is different in the two variants of our encoding.

Auxiliary steps do not influence the choice of source terms steps that are simulated. Moreover they operate on restricted channels, \ie are unobservable. Accordingly they do not change the state of the target term modulo the considered reference relations $ \barbBisim $ and $ \barbCS $. We introduce some auxiliary lemmata to support this claim.

The encoding $ \EncCI{\cdot} $ translates source term barbs $ \ch $ into free announcements with $ \Proj{\Renam{\ch}}{1} $ as first value and a lock $ \lock $ as third value that computes to $ \top $. The two coordinators, \ie outer encodings, we introduce later, restrict the free $ \act $-channel of $ \EncCI{\cdot} $.

\begin{definition}[Translated Barbs]
	Let $ T \in \procT $ such that $ \exists S\logdot \EncCI{S} \stepsT T $, $ \exists S\logdot \EncCO{S} \stepsT T $, or $ \exists S\logdot \EncDO{S} \stepsT T $.
	$ T $ has a translated barb $ \ch $, denoted by $ T\HasBarb{\EncCI{\cdot}\ch} $, if
	\begin{compactitem}
		\item there is an unguarded output $ \Out{\act}{\Proj{\Renam{\ch}}{1}, \req, \lock, \simu} $---on a free channel $ \act $ in the case of $ \EncCI{\cdot} $ or the outermost variant of $ \act $ in the case of the later introduced encodings $ \EncCO{\cdot} $ and $ \EncDO{\cdot} $---in $ T $ or
		\item such an announcement was consumed to unguard an $ \ITE{\cdot}{\cdot}{\cdot} $-construct testing $ \lock $ and this construct is still not resolved in $ T $
	\end{compactitem}
	such that all locks that are necessary to instantiate $ \lock $ are positively instantiated.
\end{definition}

Analysing the encoding function in Figure~\ref{fig:innerEncoding} we observe that guarded subterms $ S' $ of a of a source term $ S $, \eg $ S = a \rightarrow S' $, are translated into guarded subterms of $ \EncCI{S} $, whereas the translations of unguarded subterms, \eg $ S = \CSPPar{S'}{S''}{A} $, remain unguarded.

\begin{obs}
	$ $\\
	Let $ S, S' \in \procS $ such that $ S' $ is a subterm of $ S $. Then $ \EncCI{S'} $ is guarded in $ \EncCI{S} $ iff $ S' $ is guarded in $ S $.
	\label{obs:guardedSourceVsTarget}
\end{obs}

We also observe that an encoded source term has a translated barb iff the corresponding source term has the corresponding source term barb.

\begin{obs}
	For all $ S \in \procS $, it holds $ S\HasBarb{\ch} $ iff $ \EncCI{S}\HasBarb{\EncCI{\cdot}\ch} $.
	\label{obs:transBarbs}
\end{obs}

All instances of success in the translation result from success in the source. More precisely the only way to obtain $ \success $ in the translation is by $ \EncCI{\success} \deff \success $.

\begin{obs}
	For all $ S \in \procS $, it holds $ S\hasSuccess $ iff $ \EncCI{S}\hasSuccess $.
	\label{obs:success}
\end{obs}

The simplest case of a step that cannot change the state of a term modulo $ \barbBisim $, is a step on a restricted channel that is not in conflict with any other reachable step of the term.

\begin{lemma}
	Let $ T, T' \in \procT $ and $ T \stepT T' $ be a step on an unobservable channel such that no alternative step of $ T $ or its derivatives is in conflict to the step $ T \stepT T' $. Then $ T \barbBisim T' $.
	\label{lem:noConflicts}
\end{lemma}

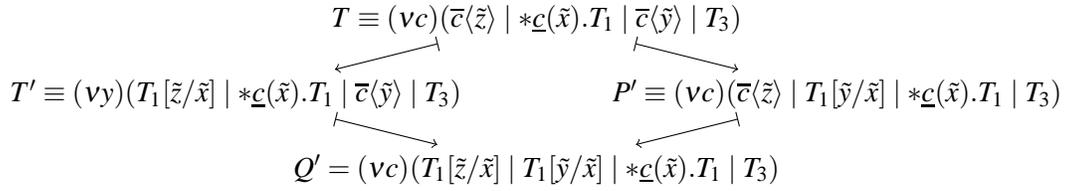
\begin{figure}[t]
	\centering
	\begin{tikzpicture}[auto]
		\node (T) at (4, 2) {$ T \equiv \Res{c}{\Out{c}{\tilde{z}} \mid \RepInput{c}{\tilde{x}}{T_1} \mid \Out{c}{\tilde{y}} \mid T_3}  $};
		\node (T') at (0, 1) {$ T' \equiv \Res{y}{T_1[ \tilde{z}/\tilde{x} ] \mid \RepInput{c}{\tilde{x}}{T_1} \mid \Out{c}{\tilde{y}} \mid T_3} $};
		\node (P') at (8, 1) {$ P' \equiv \Res{c}{\Out{c}{\tilde{z}} \mid T_1[ \tilde{y}/\tilde{x} ] \mid \RepInput{c}{\tilde{x}}{T_1} \mid T_3} $};
		\node (Q') at (4, 0) {$ Q' = \Res{c}{T_1[ \tilde{z}/\tilde{x} ] \mid T_1[ \tilde{y}/\tilde{x} ] \mid \RepInput{c}{\tilde{x}}{T_1} \mid T_3} $};
		
		\path[|->] (T) edge (T');
		\path[|->] (T) edge (P');
		\path[|->] (T') edge (Q');
		\path[|->] (P') edge (Q');
	\end{tikzpicture}
	\caption{Diamond Property.}
	\label{fig:diamondProp}
\end{figure}

\begin{proof}
	Let $ \mathcal{B} $ be the reflexive closure over the set of pairs $ \left( T, T' \right) $ such that $ T \stepT T' $ is a step on an unobservable channel and no alternative step of $ T $ or its derivatives is in conflict with the step $ T \stepT T' $. We show that $ \mathcal{B} $ is a bisimulation.
	Let $ \left( T, T' \right) \in \mathcal{B} $.
	We have to prove the following four conditions:
	\begin{compactenum}
		\item $ T \stepT P' $ implies $ \exists Q'\logdot T' \stepsT Q' \wedge P' \barbBisim Q' $:\\
			Without loss of generality assume $ T \stepT T' $ reduces $ \Out{c}{\tilde{z}} $ and $ \RepInput{c}{\tilde{x}}{T_1} $ (the case of non-replicated input is similar).
			Then $ T \equiv \Res{c}{\Out{c}{\tilde{z}} \mid \RepInput{c}{\tilde{x}}{T_1} \mid T_2} $ and $ T' \equiv \Res{y}{T_1[ \tilde{z}/\tilde{x} ] \mid \RepInput{c}{\tilde{x}}{T_1} \mid T_2} $.\\
			If $ P' = T' $ then choose $ Q' = P' $. $ Q' \stepsT Q' $ and $ P' \barbBisim Q' $ follow from reflexivity.\\
			Else if $ T \stepT P' $ is a step on $ c $, then, since there are no conflicts with $ T \stepT T' $, the two steps reduce different outputs on $ c $ but the same replicated input. Hence $ T_2 \equiv \Out{c}{\tilde{y}} \mid T_3 $ and $ P' \equiv \Res{c}{\Out{c}{\tilde{z}} \mid T_1[ \tilde{y}/\tilde{x} ] \mid \RepInput{c}{\tilde{x}}{T_1} \mid T_3} $. Then $ T' $ can perform this step such that $ T' \stepT Q' $ with $ Q' = \Res{c}{T_1[ \tilde{z}/\tilde{x} ] \mid T_1[ \tilde{y}/\tilde{x} ] \mid \RepInput{c}{\tilde{x}}{T_1} \mid T_3} $. Also $ P' $ can perform the same step as $ T \stepT T' $ such that $ P' \stepT Q' $ (compare to Figure~\ref{fig:diamondProp}). Since no alternative step of a derivative of $ T $ can be in conflict with this step, we have $ \left( P', Q' \right) \in \mathcal{B} $.\\
			Else there is $ c' \neq c $ such that $ T_2 \equiv \Out{c'}{\tilde{y}} \mid \Input{c'}{\tilde{v}}{T_1'} \mid T_3 $ and $ P' \equiv \Res{c}{\Out{c}{\tilde{z}} \mid \RepInput{c}{\tilde{x}}{T_1} \mid T_1'[ \tilde{y}/\tilde{v} ] \mid T_3} $ (the case of replicated input is similar). Again $ T' $ can perform this step such that $ T' \stepT Q' $ with $ Q' = \Res{c}{T_1[ \tilde{z}/\tilde{x} ] \mid \RepInput{c}{\tilde{x}}{T_1} \mid T_1'[ \tilde{y}/\tilde{v} ] \mid T_3} $. Also $ P' $ can perform the same step as $ T \stepT T' $ such that $ P' \stepT Q' $. Since no alternative step of a derivative of $ T $ can be in conflict with this step, we have $ \left( P', Q' \right) \in \mathcal{B} $.
		\item $ T' \stepT Q' $ implies $ \exists P'\logdot T \stepsT P' \wedge P' \barbBisim Q' $:\\
			Choose $ P' = Q' $. Then $ T \stepT T' \stepT Q' $ and, by reflexivity, $ P' \barbBisim Q' $. 
		\item $ T\reachSuccess $ iff $ T'\reachSuccess $:\\
			Once success is unguarded it cannot be removed. Accordingly the step can only add an unguarded instance of success, which then is reachable from $ T $. By 1. and 2., $ T $ and $ T' $ can reach the same occurrences of success.
		\item $ T\ReachBarb{\EncCI{\cdot}a} $ iff $ T'\ReachBarb{\EncCI{\cdot}a} $:\\
			Since there are only outputs but no inputs on the free variant of $ \act $, steps can produce but not reduces free announcements. Every free announcement introduced by $ T \stepT T' $ is also reachable in $ T $. By 1. and 2., $ T $ and $ T' $ reach the same translated barbs.
	\end{compactenum}
\end{proof}

Many auxiliary steps implement the forwarding of announcements. They are steps on restricted channels such that there is always exactly one replicated input on this channel. This ensures that these steps cannot be in conflict with other steps of the encoding and thus do not change the state modulo $ \barbBisim $.

\begin{prop}
	Let $ T, T' \in \procT $ and $ T \stepT T' $ be a step on a restricted channel $ c $ such that the only input on $ c $ in $ T $ and all derivatives of $ T $ is exactly one replicated input.
	Then $ T \barbBisim T' $.
	\label{prop:auxStepsForward}
\end{prop}

\begin{proof}
	Two steps reducing the same replicated input (but different outputs) are not in conflict with each other. Thus $ T $ cannot perform a step that is conflict with $ T \stepT T' $. Note that replicated inputs are never removed. Since this replicated input is the only input on $ c $ in all derivatives of $ T $, no alternative step of any derivative of $ T $ is in conflict with $ T \stepT T' $. By Lemma~\ref{lem:noConflicts}, then $ T \barbBisim T' $.
\end{proof}

Variants of the channels $ \req, \lreq $, and $ \rreq $ do not carry parameters. For channels like these, a conflict can only result from two steps that reduce different (replicated) inputs, because the derivatives can differ only due to different continuations of the respective inputs.

\begin{prop}
	Let $ T, T' \in \procT $ and $ T \stepT T' $ reduce a restricted $ c $ such that no value is transmitted and there is at most one input or replicated input on $ c $ in $ T $ and all derivatives of $ T $.
	Then $ T \barbBisim T' $.
	\label{prop:auxStepsReq}
\end{prop}

\begin{proof}
	Since there is always at most one (replicated) input on $ c $, alternative steps on this channel can only compete for different outputs. Let $ T \stepsT T_1 $. Since in $ \procT $ outputs have no continuation and because $ c $ does not carry a value, the continuations of two steps $ T_1 \stepT T_1' $ and $ T_1 \stepT T_1'' $ that reduce different outputs on $ c $ but the same (replicated) input are structural congruent, \ie $ T_1' \equiv T_1'' $.
	By Lemma~\ref{lem:noConflicts} and because all other steps on different channels are not in conflict with $ T \stepsT T' $, then $ T \barbBisim T' $.
\end{proof}

The encoding propagates announcements through the translated parallel structure. In the translation of parallel operators it combines all left and right announcements \wrt to the same channel name, if this channel needs to be synchronised. Therefore we copy announcements.
We use locks carrying a Boolean value to indicate whether an announcement was already used to simulate a source term step. These locks carry $ \top $ in the beginning and are swapped to $ \bot $ as soon as the announcement was used. In each state there is at most one positive instantiation of each lock and as soon as a lock is instantiated negatively it never becomes positive again.

\begin{lemma}
	Let $ T \in \procT $ such that $ \exists S\logdot \EncCI{S} \stepsT T $. Then for each variant $ l $ of the names $ \lock, \lLock, \rLock $
	\begin{compactenum}
		\item there is at most one positive instantiation of $ l $ in $ T $,
		\item if there is a positive instantiation of $ l $ in $ T $ then there is no other instantiation of $ l $ in $ T $,
		\item if there is a negative instantiation of $ l $ in $ T $ then no derivative of $ T $ contains a positive instantiation of $ l $.
	\end{compactenum}
	\label{lem:sumLocks}
\end{lemma}

\begin{proof}
	Analysing the encoding function in Figure~\ref{fig:innerEncoding} we observe that initially no instantiations of locks are unguarded. $ \Sim $ and the translation of external choice are the only parts of the translation that introduce instantiations of locks and both restrict the respective locks.
	
	In the translation of external choice all instantiations of the lock $ \lock $ are guarded by an input or a replicated input on $ \req $. Moreover, to unguard one of the later two instantiations within the $ \ITE{\cdot}{\cdot}{\cdot} $-construct, a step on $ \simu_i $ is necessary. Therefore we need an instantiation of $ \simu_i $. The only instances of a variant of $ \simu, \simu_i, \lsimu $, $ \rsimu $ are generated by $ \Sim $. There they are guarded and to unguard them a positive instantiation of the corresponding lock has to be consumed. This way only a single positive instantiation can be unguarded, but $ \RepInput{\req}{}{\Out{\lock}{\bot}} $ allows to obtain several negative instantiations of $ \lock $ if there are several outputs on $ \req $.
	
	To unguard an instantiation of $ \lock $ within the $ \ITE{\cdot}{\cdot}{\cdot} $-constructs in $ \Sim $ a step on $ \lock' $ is necessary. Initially there is only a single unguarded output on $ \lock' $. A subsequent output on $ \lock' $ can be unguarded  by consuming a negative instantiation of $ \simu $ and that requires again the consumption of a positive instantiation of $ \lock $. Moreover, if a negative instantiation of $ \lock $ is unguarded, then also $ \RepInput{\req}{}{\Out{\lock}{\bot}} $ but no output on $ \lock' $ is unguarded.
	
	Thus both cases 1.\@ and 2.\@ follow by induction over the number of steps in $ \EncCI{S} \stepsT T $. Since initially only positive instantiations of $ l $ are reachable, by 1.\@ and 2., and because the unguarding of a new positive instance of $ l $ requires the consumption of a positive instantiation of $ l $, property 3.\@ holds.
\end{proof}

Moreover each target term contains at most a single input or replicated input for each variant of $ \req, \lreq $, and $ \rreq $.

\begin{lemma}
	Let $ T \in \procT $ such that $ \exists S\logdot \EncCI{S} \stepsT T $. Then for each variant $ r $ of the names $ \req, \lreq, \rreq $ there is at most one (replicated) input on $ r $ in $ T $.
	\label{lem:reqLocks}
\end{lemma}

\begin{proof}
	(Replicated) inputs on $ r $ are introduced by $ \Sim $ and the translation of external choice.
	
	In $ \Sim $ to unguard an input on $ r $ an output on $ \lock' $ has to be consumed. Initially there is a single such output. Additional outputs on $ \lock' $ can only be unguarded by consuming a positive instantiation of the lock $ \lock $. Unguarding a positive instantiation of $ \lock $ in turn requires to consume an input on $ r $. Unguarding a replicated input on $ r $ also unguards a negative instantiation of $ \lock $.
	
	The translation of external choice initially offers exactly one unguarded input on $ r $. To unguard an additional (replicated) input on $ r $, we have to consume a positive instantiation of the lock $ \lock $ (to obtain an instantiation of $ \simu $). Unguarding a positive instantiation of $ \lock $ in turn requires to consume an input on $ r $. Unguarding a replicated input on $ r $ also unguards a negative instantiation of $ \lock $.
	
	By induction over the number of steps in $ \EncCI{S} \stepsT T $, we can show that there is at most one (replicated) input on $ r $ in $ T $.
\end{proof}

$ \Synch{c} $ combines each left announcement of this action with each right announcement of this action. Therefore each left announcement---transmitted over $ \Proj{\Renam{c}}{2} $ to keep track of the source term action---restricts its own version of $ \syn $ and $ \syn' $. Then over each $ \syn $ all right announcements---initially transmitted over $ \Proj{\Renam{c}}{3} $---are received and forwarded to the next variant of $ \syn $ by a message on $ \syn' $. Derivatives of $ \Synch{c} $ can differ in the order in that left announcements on $ \Proj{\Renam{c}}{2} $ were received. Two left announcements for the same action cannot be processed concurrently, but also the source term cannot perform two steps on the same synchronised channel concurrently. We show that the different order of left announcements does not matter.

\begin{lemma}
	Let $ T, T' \in \procT $ such that $ \exists S\logdot \EncCI{S} \stepsT T \stepT T' $ and $ T \stepT T' $ reduces an output on $ \Proj{\Renam{c}}{2} $. Moreover assume that for all steps $ T_1 \stepT T_1' $ on a variant of $ \syn, \syn', \nextSyn $ with $ \exists S'\logdot \EncCI{S'} \stepsT T_1 $ it holds $ T_1 \barbBisim T_1' $.
	Then $ T \barbBisim T' $.
	\label{lem:orderLeftAnn}
\end{lemma}

\begin{proof}
	The only part of $ \EncCI{\cdot} $ that provides inputs on $ \Proj{\Renam{c}}{2} $ is $ \Synch{c} $. Since $ \Proj{\Renam{c}}{2} $ is restricted in the translation of the parallel operator, $ T $ can have at most one unguarded input on $ \Proj{\Renam{c}}{2} $ but several outputs on $ \Proj{\Renam{c}}{2} $. Thus different steps on this channel are in conflict with each other.
	A new input on $ \Proj{\Renam{c}}{2} $ is unguarded by reducing the replicated input $ \RepIn{\nextSyn}{\syn} $. Thus the continuations of different steps on $ \Proj{\Renam{c}}{2} $ differ by the variant of $ \syn $ only. Each reduction of the input on $ \Proj{\Renam{c}}{2} $ immediately restricts new variants of $ \syn $ and $ \syn' $ and provides a new output $ \Out{\nextSyn}{\syn} $. Since all steps on variants of $ \syn, \syn', \nextSyn $ do not change the state modulo $ \barbBisim $ and because all variants of $ \syn, \syn', \nextSyn, \Proj{\Renam{c}}{2} $ are restricted, the continuations of different inputs on  $ \Proj{\Renam{c}}{2} $ cannot be distinguished by $ \barbBisim $.
	Thus $ T \barbBisim T' $.
\end{proof}

\section{The Centralised Encoding}
\label{sec:central}

Figure~\ref{fig:innerEncoding} describes how to translate CSP actions into announcements augmented with locks and how the other operators are translated to either forward or combine these announcements and locks. With that $ \EncCI{\cdot} $ provides the basic machinery of our encoding from CSP into CCS with name passing and matching. However it does not allow to simulate any source term step. Therefore we need a second (outer) layer that triggers and coordinates the simulation of source term steps. We consider two ways to implement this coordinator: a centralised and a de-centralised coordinator. The centralised coordinator is depicted in Figure~\ref{fig:centralised}.

\begin{figure}
	\begin{align*}
		\EncCO{P} \deff & \Res{\act, \once}{\EncCI{P} \mid \Out{\once}{} \mid\RepInput{\once}{}{\Input{\act}{\ch, \req, \lock, \simu}{\left( \Out{\req}{} \mid \Input{\lock}{\boolV}{\left( \Out{\once}{} \mid \IF{\boolV}\;\THEN{\Out{\simu}{\top}} \right)} \right)}}}
	\end{align*}
	\caption{A \textbf{centralised} encoding from CSP into CCS with value passing.}
	\label{fig:centralised}
\end{figure}

The channel $ \once $ is used to ensure that simulation attempts of different source term steps cannot overlap each other. For each simulation attempt exactly one announcement is consumed. The coordinator then triggers the computation of the respective lock that was transmitted in the announcement. This request for the computation of the lock is propagated along the parallel structure induced by the translations of parallel operators until---in the leafs---encodings of sums are reached. There the request for the computation yields the transmission of the current value of the respective lock. While being transmitted back to the top of the tree, different locks that refer to synchronisation in the source terms are combined. If the computation of the lock results with $ \top $ at the top of the tree, the respective source term step is simulated. Else the encoding aborts the simulation attempt and restores the consumed informations about the values of the respective locks. In both cases a new instance of $ \Out{\once}{} $ allows to start the next simulation attempt. Accordingly only some post-processing steps can overlap with a new simulation attempt.

The central coordinator respects the protocol on locks used to ensure that each announcement is only used once to simulate a source term step, \ie it preserves the properties of locks formulated in Lemma~\ref{lem:sumLocks}.

\begin{prop}
	Let $ T \in \procT $ such that $ \exists S\logdot \EncCO{S} \stepsT T $. Then for each variant $ l $ of the names $ \lock, \lLock, \rLock $
	\begin{compactenum}
		\item there is at most one positive instantiation of $ l $ in $ T $,
		\item if there is a positive instantiation of $ l $ in $ T $ then there is no other instantiation of $ l $ in $ T $, and
		\item if there is a negative instantiation of $ l $ in $ T $ then no derivative of $ T $ contains a positive instantiation of $ l $.
	\end{compactenum}
	\label{prop:sumLocksCentral}
\end{prop}

\begin{proof}
	The encoding in Figure~\ref{fig:centralised} does not introduce new instantiations of $ l $. It does provide additional instantiations of $ \simu $, but to unguard them again a positive instantiation of the corresponding lock $ \lock $ has to be consumed. Thus $ \EncCO{\cdot} $ preserves the properties of locks formulated in Lemma~\ref{lem:sumLocks}.
\end{proof}

Similarly Lemma~\ref{lem:reqLocks} is preserved.

\begin{prop}
	Let $ T \in \procT $ such that $ \exists S\logdot \EncCO{S} \stepsT T $. Then for each variant $ r $ of the names $ \req, \lreq, \rreq $ there is at most one (replicated) input on $ r $ in $ T $.
	\label{prop:reqLocksCentral}
\end{prop}

\begin{proof}
	Follows from Lemma~\ref{lem:reqLocks} and Proposition~\ref{prop:sumLocksCentral}, because the encoding in Figure~\ref{fig:centralised} provides additional instantiations of $ \simu $, but to unguard them a positive instantiation of the corresponding lock $ \lock $ has to be consumed.
\end{proof}

Lemma~\ref{lem:orderLeftAnn} is preserved by $ \EncCO{\cdot} $, because the encoding in Figure~\ref{fig:centralised} does not use variants of the names $ \syn, \syn', \nextSyn $, and $ \Proj{\Renam{c}}{2} $.

\begin{prop}
	Let $ T, T' \in \procT $ such that $ \exists S\logdot \EncCO{S} \stepsT T \stepT T' $ and $ T \stepT T' $ reduces an output on $ \Proj{\Renam{c}}{2} $. Moreover assume that for all steps $ T_1 \stepT T_1' $ on a variant of $ \syn, \syn', \nextSyn $ with $ \exists S'\logdot \EncCO{S'} \stepsT T_1 $ it holds $ T_1 \barbBisim T_1' $.
	Then $ T \barbBisim T' $.
	\label{prop:orderLeftAnnCentral}
\end{prop}

\begin{proof}
	The encoding in Figure~\ref{fig:centralised} does not use variants of the names $ \syn, \syn', \nextSyn $, and $ \Proj{\Renam{c}}{2} $. Thus this Proposition follows from Lemma~\ref{lem:orderLeftAnn}.
\end{proof}

As we prove below, the points of no return in the centralised encoding can result from the consumption of action announcements by the outer encoding in Figure~\ref{fig:centralised} if the corresponding lock computes to $ \top $. Moreover the encoding of internal choice and divergence introduces simulation steps, namely all steps on variants of the channels $ \much $, $ \rep $, and $ \renam'(X) $. All remaining steps of the centralised encoding are auxiliary.

\begin{definition}[Auxiliary and Simulation Steps]
	A step $ T \stepT T' $ such that $ \exists S \in \procS\logdot \EncCO{S} \stepsT T $ is called a \emph{simulation step}, denoted by $ T \sStepT T' $, if $ T \step T' $ is a step on the outermost channel $ \act $ and the computation of the value of the received lock $ \lock $ will return $ \top $ or it is a step on a variant of $ \much $, $ \rep $, or $ \renam'(X) $.
	
	Else the step $ T \stepT T' $ is called an \emph{auxiliary step}, denoted by $ T \aStepT T' $.
	\label{def:auxStepsCentral}
\end{definition}

\noindent
Let $ \aStepsT $ denote the reflexive and transitive closure of $ \aStepT $ and let $ \sStepsT \deff \aStepsT \sStepT \aStepsT $.
Auxiliary steps do not change the state modulo $ \barbBisim $.

\begin{lemma}
	$ T \aStepT T' $ implies $ T \barbBisim T' $ for all target terms $ T, T' $.
	\label{lem:auxStepsCentral}
\end{lemma}

\begin{proof}
	We distinguish the following cases \wrt the channel $ x $ that is reduced in the step $ T \aStepT T' $.
	\begin{compactenum}

		\item $ \boolV $ is a placeholder for $ t $ and $ f $, but, in contrast to $ t $ and $ f $, $ \boolV $ itself is never used as a channel name. Also $ \tau, \ch, \lch, \rch, z $, and $ \Proj{\Renam{\ch'}}{1} $ for all source term names $ \ch' $ are never used as channels.
		\item All variants of one of the names $ \act $ except for the outermost, $ \act', \lock', \req', \lreq', \rreq', \nextSyn, \syn, \syn' $, and $ \once $ are used as simple forwarders. If we analyse the encoding functions in Figure~\ref{fig:innerEncoding} and Figure~\ref{fig:centralised}, we observe that they are always restricted and there is exactly one replicated input and no other input on the respective variant in their scope. Thus, for all target terms $ T $ such that $ \exists S\logdot \EncCO{S} \steps T $, all steps on such channels satisfy the conditions specified by Proposition~\ref{prop:auxStepsForward}. Hence $ T \barbBisim T' $.
		\item The name $ \Proj{\Renam{\ch}}{3} $ is transmitted over $ \nextSyn $ in $ \Synch{\cdot} $ as initial value of $ \syn $. Thus, similarly to $ \syn $ because of Proposition~\ref{prop:auxStepsForward}, $ T \barbBisim T' $.
		\item The case of $ x $ being a variant of $ \req, \lreq, \rreq $ follows from Proposition~\ref{prop:auxStepsReq} and Proposition~\ref{prop:reqLocksCentral}.
		\item The case of $ x $ being a variant of $ \Proj{\Renam{\ch}}{2} $ follows from 2.\@ and Proposition~\ref{prop:orderLeftAnnCentral}.
		\item Variants of the names $ t, f $ are used to implement Boolean valued locks and an $ \ITE{\cdot}{\cdot}{\cdot} $-construct testing such locks. By Proposition~\ref{prop:sumLocksCentral}, there is at most one positive instantiation of each lock and by definition all negative instantiations of the same lock---and also positive ones---are structural congruent. Since each $ \ITE{\cdot}{\cdot}{\cdot} $-construct restricts its own variants of $ t $ and $ f $ and because there is never a positive and a negative instantiation of the same lock (Proposition~\ref{prop:sumLocksCentral}), all conflicts between two steps on variants of $ t $ and $ f $ result into structural congruent continuations and a step on variants of $ t $ and $ f $ cannot be in conflict with any other step on a different channel of $ T $ or its derivatives. Because $ \equiv \; \subseteq \; \barbBisim $ and by Lemma~\ref{lem:noConflicts}, then $ T \barbBisim T' $.
		\item Variants of the names $ \lock, \lLock, \rLock $ refer to Boolean valued locks. In the centralised encoding all announcements are propagated upwards---and on their way upwards some of them are composed---until they reach the outer layer $ \EncCO{\cdot} $. $ \once $ ensures that only a single announcement is processed at a time. A new output on $ \once $ can only be unguarded by consuming an instantiation of the lock $ \lock $ of the announcement that is currently processed. After the consumption of an announcement the output $ \Out{\req}{} $ triggers the computation of an instantiation of $ \lock $. Technically an instantiation of a lock is an output on $ \lock $ and the corresponding inputs are part of the $ \ITE{\cdot}{\cdot}{\cdot} $-construct testing this value. $ \Sim $ is the only part of the encoding that introduces such $ \ITE{\cdot}{\cdot}{\cdot} $-constructs for variants of $ \lock, \lLock, \rLock $ and there the $ \ITE{\cdot}{\cdot}{\cdot} $-constructs are guarded by an input on $ \req $. Each step on $ \req $ unguards a nested $ \ITE{\cdot}{\cdot}{\cdot} $-construct testing a variant of $ \lLock $ and $ \rLock $. Since outputs on variants of $ \req, \lreq, \rreq $ move downwards along the translation of the parallel tree of the source term and because of $ \once $, no two different $ \ITE{\cdot}{\cdot}{\cdot} $-constructs for the same lock are ever unguarded. By Proposition~\ref{prop:sumLocksCentral}, there is at most one positive instantiation of each lock in $ T $ and if there is a positive instantiation then there is no negative instantiation of the same lock. Thus a step reducing a positive instantiation of a lock cannot be in conflict with any other step of $ T $ or derivatives of $ T $. By Lemma~\ref{lem:noConflicts}, then $ T \barbBisim T' $.\\
			By definition, all negative instantiations of the same lock are structural congruent. Thus, since there is only a single $ \ITE{\cdot}{\cdot}{\cdot} $-construct, two alternative steps that reduce different negative instantiations of the same lock result into structural congruent derivatives. All steps on other channels cannot be in conflict with a step reducing a negative instantiation of the respective lock. Because $ \equiv \; \subseteq \; \barbBisim $ and by Lemma~\ref{lem:noConflicts}, then $ T \barbBisim T' $.
		\item In the case of $ x $ being the outermost variant of $ \act $, Definition~\ref{def:auxStepsCentral} ensures that the lock $ \lock $ received in this step will compute to $ \bot $. By induction on the parallel structure of the respective source term, we show that the encoding then ensures that all instantiations of locks that were consumed to compute the instantiation of $ \lock $ are restored with the same truth value. This holds, because $ \Sim $ ensures that each combination of a positive instantiated lock from the left and a negative instantiated lock from the right causes an output $ \Out{\lsimu}{\bot} $. This output is propagated downwards and causes the outputs $ \Out{\lsimu}{\bot} $ and $ \Out{\rsimu}{\bot} $ for each pair of positive instantiated left and right locks combined below. In the translation of external choice these outputs on variants of $ \simu, \lsimu, \rsimu $ cause the unguarding of a fresh positive instantiation of the respective lock. Negative instantiations do not need to be restored, because they are introduced by $ \RepInput{\req}{}{\Out{\lock}{\bot}} $ that provides as many negative instantiations as there are requests $ \Out{\req}{} $ for them.
			Also, only if the lock computes to $ \top $ a positive instantiation of $ \simu $ is unguarded and propagated downwards. Since positive instantiations of variants of $ \simu, \lsimu, \rsimu $ are the only way to unguard an encoded source term continuation in the translation of external choice, a step reducing an announcement such that the respective lock will be computed to $ \bot $ cannot influence reachability of barbs or success.
			Thus modulo some auxiliary steps considered above, \ie modulo steps that do not change the state of the term modulo $ \barbBisim $, in the present case $ T $ and $ T' $ differ by the consumption of the respective announcement only. Since announcements are not success, are not observable, and, because of the negative lock, this announcement is not a translated barb, the difference between $ T $ and $ T' $ is not observable by $ \barbBisim $, \ie $ T \barbBisim T' $.
	\end{compactenum}
\end{proof}

By distinguishing auxiliary and simulation steps, we can prove a condition stronger than operational correspondence, namely that each source term step is simulated by exactly one simulation step.

\begin{lemma}
	$ \forall S, S'\logdot S \stepS S' $ iff $ \exists T\logdot \EncCO{S} \sStepsT T \wedge \EncCO{S'} \barbBisim T $.
	\label{lem:sourceVsSimStep}
\end{lemma}

\begin{proof}
	Let $ S, S' \in \procS $.
	\begin{compactitem}
		\item[`if'-part:] Assume $ S \stepS S' $. Then either there is some source term name $ \ch $ such that $ S \Lstep{\ch} S' $ or $ S \Lstep{\tau} S' $. In the first case at least one action prefix $ \ch \rightarrow \cdot $ is reduced. The second case results from divergence, internal choice, concealment, or recursion.
			\begin{compactenum}
				\item The encoding translates action prefixes $ \ch \rightarrow \cdot $ into announcements with $ \Proj{\Renam{\ch}}{1} $ as first value. By Observation~\ref{obs:guardedSourceVsTarget} and because the outer encoding in Figure~\ref{fig:centralised} does not guard $ \EncCI{P} $, these announcements are unguarded for all source term action prefixes that are reduced in $ S \Lstep{\ch} S' $. By induction on the structure of the source term, we show that these announcements can be transferred all the way up in the parallel tree and are combined along this way---for each parallel operator that synchronises two $ \ch $-actions in the source, two announcements are combined in the translation---such that a single announcement for this action reaches the outermost $ \act $-channel. The coordinator performs a step on $ \once $ and then receives this announcement and requests the computation of the lock by sending $ \Out{\req}{} $. Since initially all locks are instantiated positive this computation results $ \top $. As a consequence $ \Out{\simu}{\top} $ is propagated downwards and ensures that the encodings of all source term continuations that are unguarded by $ S \Lstep{\ch} S' $ can be unguarded by auxiliary steps in the translation.
					Moreover $ \Out{\simu}{\top} $ ensures that the consumed instantiations of locks can only be re-instantiated with the value $ \bot $. Let $ T $ denote the result of this simulation.
					
					The negative instantiations of the locks ensure that no step of $ S $ that is in conflict with $ S \stepS S' $ can be simulated by $ T $ and removes translated barbs that refer to barbs removed by $ S \stepS S' $.
					The only non-auxiliary step in the simulation $ \EncCO{S} \stepsT T $ is the simulation step that consumes the announcement on the top of the tree on the outermost $ \act $-channel, \ie $ \EncCO{S} \sStepsT T $. With Observation~\ref{obs:transBarbs}, Observation~\ref{obs:success}, and because in the end of the simulation the encodings of the respective source term continuations are unguarded, $ T $ and $ \EncCO{S'} $ have the same ability to reach success and reach the same translated observables. Hence $ \EncCO{S'} \barbBisim T $.
				\item Divergence is translated into the divergent target term $ \Res{\rep}{\Out{\rep}{} \mid \RepInput{\rep}{}{\Out{\rep}{}}} $. By Observation~\ref{obs:guardedSourceVsTarget}, simulating $ S \stepS S' $ in this case requires only a single simulation step on the respective variant of $ \rep $. Let $ T $ be the derivative of this step. Since the steps on $ \rep $ are not observable modulo $ \barbBisim $ in this case, we have $ \EncCO{S} \barbBisim T \barbBisim \EncCO{S'} $.
				\item Internal choice $ \inchoice{P}{Q} $ is translated into $ \Res{\much}{\Input{\much}{}{\EncCI{P}} \mid \Input{\much}{}{\EncCI{Q}} \mid \Out{\much}{}} $. By Observation~\ref{obs:guardedSourceVsTarget}, simulating $ S \stepS S' $ in this case requires only a single simulation step on the respective variant of $ \much $. Let this step unguard $ \EncCI{P} $ if $ S \stepsS S' $ unguards $ P $ and else unguard $ \EncCI{Q} $. Let $ T $ be the derivative of this step. With Observation~\ref{obs:transBarbs}, Observation~\ref{obs:success}, and because the simulation unguards the encoding of the respective source term continuation, $ T $ and $ \EncCO{S'} $ have the same ability to reach success and reach the same translated observables. Hence $ \EncCO{S'} \barbBisim T $.
				\item In the case of concealment the source term hides a former observable action that is simulated as in 1. The translation of concealment only adds a restriction on $ \ch $ and renames the first value of the announcement into $ \tau $ such that it is never synchronised afterwards. Thus the simulation of $ S \Lstep{\tau} S' $ in this case is similar to 1.\@ except for the steps to forward the announcement within the translation of concealment.
				\item $ \mu X \cdot P $ is translated into $ \Res{\renam'(X)}{\Out{\renam'(X)}{} \mid \RepInput{\renam'(X)}{}{\EncCI{P}}} $ and $ \EncCI{X} = \Out{\renam'(X)}{} $. By Observation~\ref{obs:guardedSourceVsTarget}, simulating $ S \stepS S' $ in this case requires only a single simulation step on the respective variant of $ \renam'(X) $. This step unguards an instance of $ \EncCI{P} $. Let $ T $ be the derivative of this step. With Observation~\ref{obs:transBarbs}, Observation~\ref{obs:success}, and because the simulation unguards the encoding of the respective source term continuation, $ T $ and $ \EncCO{S'} $ have the same ability to reach success and reach the same translated observables. Hence $ \EncCO{S'} \barbBisim T $.
			\end{compactenum}
			Thus, by induction on the structure of $ S $, the encoding $ \EncCO{\cdot} $ can simulate each source term step $ S \stepS S' $ such that $ \exists T\logdot \EncCO{S} \sStepsT T \wedge \EncCO{S'} \barbBisim T $.
		\item[`only-if'-part:] Assume $ T $ such that $ \EncCO{S} \sStepsT T $ and $ \EncCO{S} \barbBisim T $. By Lemma~\ref{lem:auxStepsCentral}, it suffices to concentrate on the single simulation step in $ \EncCO{S} \sStepsT T $. In $ \EncCO{S} \sStepsT T $ either exactly one announcement \wrt to a positive lock is reduced by the simulation step (1.), or there is exactly one step---namely the simulation step---on a variant of either $ \much $ (2.), $ \rep $ (3.), or $ \renam'(X) $ (4.).
			\begin{compactenum}
				\item Since $ \EncCO{S} \sStepsT T $ neither contains steps on variants of $ \rep $ nor $ \much $ nor $ \renam'(X) $, no encoded source term continuation in the translation of internal choice or recursion is unguarded. Let $ T', T'' $ such that $ \EncCO{S} \aStepsT T' \sStepT T'' \aStepsT T $. $ T' \sStepT T'' $ reduces an announcement $ \Out{\act}{\ch, \req, \lock, \simu} $ such that the computation of $ \lock $ in $ T'' $ will result $ \top $.	By analysing the way the lock $ \lock $ is computed in $ T' $ we can conclude on the source term prefixes and the part of the source term parallel structure that is reflected by this simulation of a source term step. Analysing the way of the announcement we can also determine whether a source term concealment was involved.
					Because auxiliary steps cannot unguard encoded source term continuations and by Observation~\ref{obs:guardedSourceVsTarget}, then we can conclude on the structure of $ S $ and construct subject to $ S $ a source term $ S' $ such that $ S \stepS S' $ and $ S' $ results from $ S $ by reducing all action prefixes whose translation are identified by the above analyse of the way the lock $ \lock $ is computed. In the $ S' $ the respective source term continuations are unguarded. In $ T'' $ only auxiliary steps are necessary to unguard the translation of these source term continuations.
					With Lemma~\ref{lem:auxStepsCentral} and because the simulation step simulates all observable effects of the step $ S \stepsS S' $, then $ \EncCO{S'} \barbBisim T $.
				\item Since no announcements \wrt positive instantiated locks are reduced in $ \EncCO{S} \sStepsT T $, no translated barb are removed and no encoded source term continuation in the translation of external choice is unguarded. Since there is no step on a variant of $ \renam'(X) $, no encoded source term continuation in the translation of recursion is unguarded. Instead exactly one source term encoding---without loss of generality let us call this encoded source term $ \EncCI{P} $---due to the translation of internal choice is unguarded. This step ensures the respective other encoded source term alternative of the internal choice can never be unguarded, \ie is modulo $ \barbBisim $ similar to $ \Null $. This is the only effect of the steps $ \EncCO{S} \sStepsT T $ that can be observed modulo $ \barbBisim $. Therefore this internal choice translation has to be unguarded in $ \EncCO{S} $, because auxiliary steps cannot unguard encoded source term continuations. By Observation~\ref{obs:guardedSourceVsTarget}, then $ S $ contains an unguarded internal choice with $ P $ as one of the alternatives. Then $ S \stepS S' $ such that this step resolves the internal choice and unguards $ P $. With Lemma~\ref{lem:auxStepsCentral} and because the simulation step simulates all observable effects of the step $ S \stepsS S' $, then $ \EncCO{S'} \barbBisim T $.
				\item Since no announcements \wrt positive instantiated locks are reduced in $ \EncCO{S} \sStepsT T $, no translated barb are removed and no encoded source term continuation in the translation of external choice is unguarded. Since there is no step on a variant of $ \much $, no encoded source term continuation in the translation of internal choice is unguarded. Since the simulation step reduces a variant of $ \rep $, we have $ \EncCO{S} \barbBisim T $. Moreover, in this case, $ \EncCO{\Div} $ is unguarded in $ \EncCO{S} $, because auxiliary steps cannot unguard encoded source term continuations. By Observation~\ref{obs:guardedSourceVsTarget}, then $ \Div $ is unguarded in $ S $. Then $ S \stepS S' $ such that this step reduces $ \Div $. With Lemma~\ref{lem:auxStepsCentral} and because the simulation step simulates all observable effects of the step $ S \stepsS S' $, then $ \EncCO{S'} \barbBisim T $.
				\item Since no announcements \wrt positive instantiated locks are reduced in $ \EncCO{S} \sStepsT T $, no translated barb are removed and no encoded source term continuation in the translation of external choice is unguarded. Since there is no step on a variant of $ \much $, no encoded source term continuation in the translation of internal choice is unguarded. Instead exactly one source term encoding---without loss of generality let us call this encoded source term $ \EncCI{P} $---due to the translation of recursion is unguarded. This is the only effect of the steps $ \EncCO{S} \sStepsT T $ that can be observed modulo $ \barbBisim $. Therefore $ \EncCO{\mu X \cdot P} $ is unguarded in $ \EncCO{S} $, because auxiliary steps cannot unguard encoded source term continuations. By Observation~\ref{obs:guardedSourceVsTarget}, then $ \mu X \cdot P $ is unguarded in $ S $. Then $ S \stepS S' $ such that this step unfolds recursion and unguards $ P $. With Lemma~\ref{lem:auxStepsCentral} and because the simulation step simulates all observable effects of the step $ S \stepsS S' $, then $ \EncCO{S'} \barbBisim T $.
			\end{compactenum}
	\end{compactitem}
\end{proof}

\noindent
This direct correspondence between source term steps and the points of no return of their translation allows us to prove a variant of operational correspondence that is significantly stricter than the variant proposed in \cite{gorla10}.

\begin{definition}[Operational Correspondence]
	$ $\\
	An encoding $ \enc: \procS \to \procT $ is \emph{operationally corresponding} \wrt $ \barbBisim \; \subseteq \procT^2 $ if it is:
	\begin{compactitem}
		\item[\; Complete:] $ \forall S, S' \logdot S \stepsS S' $ implies $ \exists T \logdot \EncCO{S} \stepsT T \wedge \EncCO{S'} \barbBisim T $
		\item[\; Sound:] $ \forall S, T \logdot \EncCO{S} \stepsT T $ implies $ \exists S' \logdot S \stepsS S' \wedge \EncCO{S'} \barbBisim T $
	\end{compactitem}
\end{definition}

\noindent
The `if'-part of Lemma~\ref{lem:sourceVsSimStep} implies operational completeness \wrt $ \barbBisim $ and the `only-if'-part contains the main argument for operational soundness \wrt $ \barbBisim $. Hence $ \encCO $ is operational corresponding \wrt to $ \barbBisim $.

\begin{theorem}
	The encoding $ \encCO $ is operational corresponding \wrt to $ \barbBisim $.
	\label{thm:operationalCorrespondenceCentral}
\end{theorem}

\begin{proof}
	Completeness---$ \forall S, S' \logdot S \stepsS S' $ implies $ \exists T \logdot \EncCO{S} \stepsT T \wedge \EncCO{S'} \approx T $---follows from the `if'-part of Lemma~\ref{lem:sourceVsSimStep} and an induction on the number of steps in $ S \stepsS S' $.
	
	Soundness---$ \forall S, T \logdot \EncCO{S} \stepsT T $ implies $ \exists S' \logdot S \stepsS S' \wedge \EncCO{S'} \approx T $---follows from Lemma~\ref{lem:auxStepsCentral}, the `only-if'-part of Lemma~\ref{lem:sourceVsSimStep}, and an induction on the number of simulation steps in $ \EncCO{S} \stepsT T $.
\end{proof}

To obtain divergence reflection we show that there is no infinite sequence of only auxiliary steps.

\begin{lemma}
	The number of steps between two simulation steps is finite.
	\label{lem:numberStepsCentral}
\end{lemma}

\begin{proof}
	Let $ T $ be such that $ \exists S\logdot \EncCO{S} \stepsT T $.
	There are only finitely many unguarded translations of encodings of source term operators in $ T $. Let $ T' $ be the result of unguarding all translations of source term parts that can be unguarded using only auxiliary steps in $ T $. By induction on the number of simulation steps in $ \EncCO{S} \stepsT T $ the number of such auxiliary steps is finite. Since we consider only sequences $ T \aStepsT T' \aStepsT \ldots $ without simulation steps, no derivative of $ T' $ in this sequence can unguard additional translations of source term operators. The binary tree that results from the nesting of unguarded translations of parallel operator encodings in $ T' $ and its derivatives is denoted as parallel tree in the following. Auxiliary steps are steps on the following kinds of channels:
	\begin{compactenum}
		\item Since we consider only sequences $ T \aStepsT T' \aStepsT \ldots $ without simulation steps, there is at most one step on $ \once $ in this sequence.
		\item Steps on variants of $ \act, \act' $ are used to propagate announcements through the parallel tree. Since this tree is finite and because the encoding introduces one announcement per action prefix, there are only finitely many announcements in the leafs of the parallel tree. Announcements are only propagated upwards to surrounding translations of concealment and parallel operators (of which there are only finitely many). Within the nodes of the parallel tree announcements from the left and announcements from the right are combined using variants of $ \Proj{\Renam{\ch}}{2}, \Proj{\Renam{\ch}}{3}, \syn, \syn' $ and, to unguard inputs on such channels, steps on variants of $ \nextSyn $ are used. By induction over the depth of the binary tree, we show that there are always only finitely many announcements from the left and finitely many announcements from the right and thus their combinations are performed by finitely many steps. Accordingly, $ T \aStepsT T' \aStepsT \ldots $ contains only finitely many steps on variants of $ \act, \act', \Proj{\Renam{\ch}}{2}, \Proj{\Renam{\ch}}{3}, \syn, \syn', \nextSyn $.
		\item Steps on variants of $ \req, \lreq, \rreq $ are used to trigger the computation of locks. Since we consider only sequences $ T \aStepsT T' \aStepsT \ldots $ without simulation steps, there is at most one request $ \Out{\req}{} $ proposed by the coordinator in this sequence. Additionally $ T $ and $ T' $ can already contain unguarded requests, but only finitely many. The request $ \Out{\req}{} $ from the top of the parallel tree is propagated downwards by pushing one or two more such requests (in some nodes) on variants of $ \lreq, \rreq $ for each consumed request. Since the depth of the parallel tree is finite, $ T \aStepsT T' \aStepsT \ldots $ contains only finitely many steps on variants of $ \req, \lreq, \rreq $.
		\item Steps on variants of $ \lock, \lLock, \rLock, \lock', t, f $ are used to implement and test Boolean valued locks. For each step on variants of $ \req, \lreq, \rreq $ only a single instantiation of a lock can be consumed. By 3., there are only finitely many such steps. Additionally $ T $ and $ T' $ can already contain unguarded instantiations of locks and $ \ITE{\cdot}{\cdot}{\cdot} $-constructs, but only finitely many. Since each consumption of a single instantiation of a lock and its test in a $ \ITE{\cdot}{\cdot}{\cdot} $-construct requires only finitely many steps, $ T \aStepsT T' \aStepsT \ldots $ contains only finitely many steps on variants of $ \lock, \lLock, \rLock, \lock', t, f $.
		\item $ T $ and $ T' $ can only contain finitely many unguarded outputs on variants of $ \simu, \simu_i, \lsimu, \rsimu $. Additional outputs on variants of $ \simu, \simu_i, \lsimu, \rsimu $ can only be unguarded by testing the value of a lock. By 4., there are only finitely many tests of locks in $ T \aStepsT T' \aStepsT \ldots $. Thus there are only finitely many steps on variants of $ \simu, \simu_i, \lsimu, \rsimu $.
	\end{compactenum}
	Thus no sequence of auxiliary steps of $ T $ is infinite.
\end{proof}

Then divergence reflection follows from the combination of the above Lemma and Lemma~\ref{lem:sourceVsSimStep}.

\begin{theorem}
	The encoding $ \encCO $ reflects divergence.
	\label{thm:divergenceReflectionCentral}
\end{theorem}

\begin{proof}
	If $ \EncCO{S} $ is divergent then, by Lemma~\ref{lem:numberStepsCentral}, $ \EncCO{S} $ can perform an infinite sequence of steps containing infinitely many simulation steps. With Lemma~\ref{lem:sourceVsSimStep}, then $ S $ is divergent.
\end{proof}

The encoding function ensures that $ \EncCO{S} $ has an unguarded occurrence of $ \success $ iff $ S $ has such an unguarded occurrence. Operational correspondence ensures that $ S $ and $ \EncCO{S} $ also answer the question for the reachability of $ \success $ in the same way.

\begin{theorem}
	The encoding $ \encCO $ is success sensitive.
	\label{thm:successSensitivenessCentral}
\end{theorem}

\begin{proof}
	From Observation~\ref{obs:success} and Figure~\ref{fig:centralised}, $ S\hasSuccess $ iff $ \EncCO{S}\hasSuccess $. With Theorem~\ref{thm:operationalCorrespondenceCentral} and because $ \barbBisim $ respects $ \success $, then $ S\reachSuccess $ iff $ \EncCO{S}\reachSuccess $.
\end{proof}

In a similar way we can prove that a source term reaches a barb iff its translation reaches the respective translated barb.

\begin{theorem}
	$ \forall S, \ch\logdot S\ReachBarb{\ch} $ iff $ \EncCO{S}\ReachBarb{\EncCI{\cdot}\ch} $
	\label{thm:respectsBarbsCentral}
\end{theorem}

\begin{proof}
	From Observation~\ref{obs:transBarbs} and Figure~\ref{fig:centralised}, $ S\HasBarb{\ch} $ iff $ \EncCO{S}\HasBarb{\EncCI{\cdot}\ch} $. With Lemma~\ref{lem:sourceVsSimStep} and because $ \barbBisim $ respects translated barbs, then $ S\ReachBarb{\ch} $ iff $ \EncCO{S}\ReachBarb{\EncCI{\cdot}\ch} $.
\end{proof}

As proved in \cite{petersGlabbeek15}, Theorem~\ref{thm:operationalCorrespondenceCentral}, the fact that $ \barbBisim $ is success sensitive and respects (translated) barbs, Theorem~\ref{thm:successSensitivenessCentral}, and Theorem~\ref{thm:respectsBarbsCentral} imply that for all $ S $ it holds $ S $ and $ \EncCO{S} $ are (success sensitive, (translated) barb respecting, weak, reduction) bisimilar, \ie $ S \barbBisim \EncCO{S} $.
Bisimilarity is a strong relation between source terms and their translation. On the other hand, because of efficiency, distributability preserving encodings are more interesting.
Because of $ \once $ the encoding $ \encCO $ obviously does not preserves distributability. As discussed in \cite{parrowCoupled92} bisimulation often forbids for distributed encodings. Instead they propose coupled simulation as relation that still provides a strong connection between source terms and their translations but is more flexible. Following the approach in \cite{parrowCoupled92} we consider a de-centralised coordinator next.

\section{The De-Centralised Encoding}
\label{sec:decentral}

\begin{figure}
	\begin{align*}
		\EncDO{P} \deff & \Res{\act}{\EncCI{P} \mid \RepInput{\act}{\ch, \req, \lock, \simu}{\left( \Out{\req}{} \mid \Input{\lock}{\boolV}{\IF{\boolV}\;\THEN{\Out{\simu}{\top}}} \right)}}
	\end{align*}
	\caption{A \textbf{de-centralised} encoding from CSP into CCS with value passing.}
	\label{fig:decentralised}
\end{figure}

Figure~\ref{fig:decentralised} presents a de-centralised variant of the coordinator in Figure~\ref{fig:centralised}.
The only difference between the centralised and the de-centralised version of the coordinator is that the latter can request to check different locks concurrently. Technically $ \encCO $ and $ \encDO $ differ only by the use of $ \once $. As a consequence the steps of different simulation attempts can overlap and even (pre-processing) steps of simulations of conflicting source term steps can interleave to a certain degree. Because of this effect, $ \encDO $ does not satisfy the version of operational correspondence used above for $ \encCO $, but $ \encDO $ satisfies weak operational correspondence that was proposed in \cite{gorla10} as part of a set of quality criteria.

Similar to the central coordinator, the de-central coordinator respects the protocol on locks used to ensure that each announcement is only used once to simulate a source term step, \ie it preserves the properties of locks formulated in Lemma~\ref{lem:sumLocks}.

\begin{prop}
	Let $ T \in \procT $ such that $ \exists S\logdot \EncCO{S} \stepsT T $. Then for each variant $ l $ of the names $ \lock, \lLock, \rLock $
	\begin{compactenum}
		\item there is at most one positive instantiation of $ l $ in $ T $,
		\item if there is a positive instantiation of $ l $ in $ T $ then there is no other instantiation of $ l $ in $ T $, and
		\item if there is a negative instantiation of $ l $ in $ T $ then no derivative of $ T $ contains a positive instantiation of $ l $.
	\end{compactenum}
	\label{prop:sumLocksDecentral}
\end{prop}

\begin{proof}
	The encoding in Figure~\ref{fig:decentralised} does not introduce new instantiations of $ l $. It does provide additional instantiations of $ \simu $, but to unguard them a positive instantiation of the corresponding lock $ \lock $ has to be consumed. Thus $ \EncDO{\cdot} $ preserves the properties of locks formulated in Lemma~\ref{lem:sumLocks}.
\end{proof}

Similarly Lemma~\ref{lem:reqLocks} is preserved.

\begin{prop}
	Let $ T \in \procT $ such that $ \exists S\logdot \EncDO{S} \stepsT T $. Then for each variant $ r $ of the names $ \req, \lreq, \rreq $ there is at most one (replicated) input on $ r $ in $ T $.
	\label{prop:reqLocksDecentral}
\end{prop}

\begin{proof}
	Follows from Lemma~\ref{lem:reqLocks} and Proposition~\ref{prop:sumLocksDecentral}, because the encoding in Figure~\ref{fig:decentralised} provides additional instantiations of $ \simu $, but to unguard them a positive instantiation of the corresponding lock $ \lock $ has to be consumed.
\end{proof}

The encoding in Figure~\ref{fig:decentralised} does not use variants of the names $ \syn, \syn', \nextSyn $, and $ \Proj{\Renam{c}}{2} $. Because of that, Lemma~\ref{lem:orderLeftAnn} is preserved by $ \EncDO{\cdot} $.

\begin{prop}
	Let $ T, T' \in \procT $ such that $ \exists S\logdot \EncDO{S} \stepsT T \stepT T' $ and $ T \stepT T' $ reduces an output on $ \Proj{\Renam{c}}{2} $. Moreover assume that for all steps $ T_1 \stepT T_1' $ on a variant of $ \syn, \syn', \nextSyn $ with $ \exists S'\logdot \EncDO{S'} \stepsT T_1 $ it holds $ T_1 \barbBisim T_1' $.
	Then $ T \barbBisim T' $.
	\label{prop:orderLeftAnnDecentral}
\end{prop}

\begin{proof}
	The encoding in Figure~\ref{fig:decentralised} does not use variants of the names $ \syn, \syn', \nextSyn $, and $ \Proj{\Renam{c}}{2} $. Thus this Proposition follows from Lemma~\ref{lem:orderLeftAnn}.
\end{proof}

Since several announcements can be processed concurrently by the de-central coordinator, here all consumptions of announcements are auxiliary steps. Instead the consumption of positive instantiations of locks can mark a point of no return. In contrast to $ \encCO $ not every point of no return in $ \encDO $ unambiguously marks a simulation of a single source term step, because in contrast to $ \encCO $ the encoding $ \encDO $ introduces \emph{partial commitments} \cite{peters12,petersNestmann12}.

Consider the example $ E = \CSPPar{\left( \exchoice{o \to P_1}{p \to P_2} \right)}{\left( \exchoice{o \to P_3}{\exchoice{p \to P_4}{q \to P_5}} \right)}{\Set{ o, p }} $.

\noindent
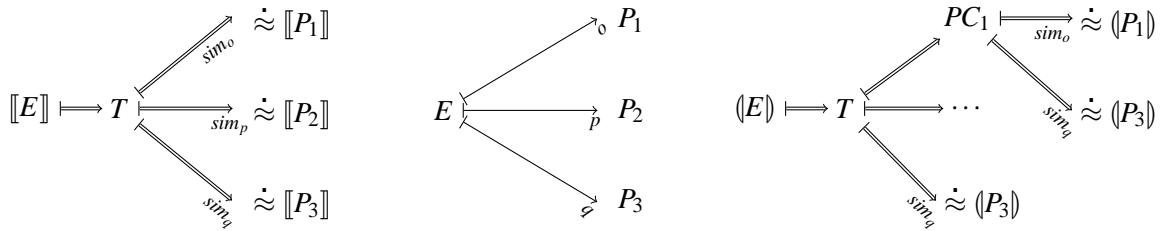
  \begin{minipage}[c]{0.3\textwidth-2pt}
      \begin{tikzpicture}[auto,node distance=1.2cm]
        \node (E)                        {$ \EncCO{E} $};
        \node (T)    [right of=E]        {$ T $};
        \node (T2)    [right=of T]        {$\hspace{1ex} \barbBisim \EncCO{P_2} $};
        \node (T1)    [above of=T2]        {$\hspace{1ex} \barbBisim \EncCO{P_1} $};
        \node (T3)    [below of=T2]        {$\hspace{1ex} \barbBisim \EncCO{P_3} $};

        \draw[|->,shorten >=-1pt,shorten <=-0.5pt] (E) -- (T);
        \draw[double] (E) -- (T);
        \draw[|->,shorten >=-1pt,shorten <=-0.5pt] (T) -- (T1.west) node [near end, below, rotate=40, scale = 0.7] {$\emph{sim}_o$};
        \draw[double] (T) -- (T1.west);
        \draw[|->,shorten >=-1pt,shorten <=-0.5pt] (T) -- (T2.west) node [at end, below, scale = 0.7] {$\emph{sim}_p$};
        \draw[double] (T) -- (T2.west);
        \draw[|->,shorten >=-1pt,shorten <=-0.5pt] (T) -- (T3.west) node [at end, below, rotate=-40, scale = 0.7] {$\emph{sim}_q$};
        \draw[double] (T) -- (T3.west);
      \end{tikzpicture}
  \end{minipage}
  \begin{minipage}[c]{0.3\textwidth-2pt}
    \begin{center}
      \begin{tikzpicture}[auto,node distance=1.2cm]
        \node (T)                        {$ E $};
        \node (d1)    [right of=T]        {};
        \node (T2)    [right of=d1]        {$\hspace{1ex} P_2$};
        \node (T1)    [above of=T2]        {$\hspace{1ex} P_1$};
        \node (T3)    [below of=T2]        {$\hspace{1ex} P_3$};

        \draw[|->,shorten >=-1pt,shorten <=-0.5pt] (T) -- (T1.west) node [at end, below, rotate=40, scale = 0.7] {$o$};
        \draw[|->,shorten >=-1pt,shorten <=-0.5pt] (T) -- (T2.west) node [at end, below, scale = 0.7] {$p$};
        \draw[|->,shorten >=-1pt,shorten <=-0.5pt] (T) -- (T3.west) node [at end, below,rotate=-40, scale = 0.7] {$q$};
      \end{tikzpicture}
      \end{center}
  \end{minipage}
  \begin{minipage}[c]{0.4\textwidth-2pt}
      \begin{tikzpicture}[auto,node distance=1.2cm]
        \node (E)                        {$ \EncDO{E} $};
        \node (T)    [right of=E]        {$ T $};
        \node (P2)    [right=1cm of T]        {$\cdots$};
        \node (P1)    [above of=P2]        {$PC_1$};
        \node (T12)    [right=1cm of P2]        {$\barbBisim \EncDO{P_{3}}$};
        \node (T11)    [above of=T12]        {$\barbBisim \EncDO{P_{1}}$};
        \node (T3)    [below of=P2]        {$\makebox[\widthof{$PC_{1}$}][l]{$\barbBisim \EncDO{P_{3}}$}$};

        \draw[|->,shorten >=-1pt,shorten <=-0.5pt] (E) -- (T);
        \draw[double] (E) -- (T);
        \draw[|->,shorten >=-1pt,shorten <=-0.5pt] (T) -- (P1);
        \draw[double] (T) -- (P1);
        \draw[|->,shorten >=-1pt,shorten <=-0.5pt] (T) -- (P2);
        \draw[double] (T) -- (P2);
        \draw[|->,shorten >=-1pt,shorten <=-0.5pt] (P1) -- (T11.west) node [near end, below, scale = 0.7] {$\emph{sim}_o$};
        \draw[double] (P1) -- (T11.west);
        \draw[|->,shorten >=-1pt,shorten <=-0.5pt] (P1) -- (T12.west) node [at end, below, rotate=-40, scale = 0.7] {$\emph{sim}_q$};
        \draw[double] (P1) -- (T12.west);
        \draw[|->,shorten >=-1pt,shorten <=-0.5pt] (T) -- (T3.west) node [at end, below, rotate=-43, scale = 0.7] {$\emph{sim}_q$};
        \draw[double] (T) -- (T3.west);
      \end{tikzpicture}
  \end{minipage}

  In the example, two sides of a parallel operator have to synchronise on either action $p$, or action $o$, or action $q$ happens without synchronisation.
  In the centralised encoding $\EncCO{E}$ the use of $ \once $ ensures that different simulation attempts cannot overlap. Thus, only after finishing the simulation of a source term step, the simulation of another source term step can be invoked. As a consequence each state reachable from encoded source terms can unambiguously be mapped to a single state of the source term. This allows us to use a stronger version of operational correspondence and, thus, to prove that source terms and their translations are bisimilar. The corresponding 1-to-1 correspondence between source terms and their translations is visualised by the first two graphs above, where $ T \barbBisim \EncCO{E} $.

  The de-centralised encoding $\EncDO{E}$ introduces partial commitments.
  Assume the translation of a source term that offers several alternative ways to be reduced. Then some encodings---as our de-central one---do not always decide on which of the source term steps should be simulated next. More precisely a partial commitment refers to a state reachable from the translation of a source term in that already some possible simulations of source term steps are ruled out, but there is still more than a single possibility left.

  In the de-centralised encoding announcements can be processed concurrently and parts of different simulation attempts can interleave. The only blocking part of the decentralised encoding are conflicting attempts to consume the same positive instantiation of a lock.
  In the presented example above there are two locks; one for each side of the parallel operator. The simulations of the step on $ o $ and $ p $ need both of these locks, whereas to simulate the step on $ q $ only a positive instantiation of the right lock needs to be consumed.
  By consuming the positive instantiation of the left lock in an attempt to simulate the step on $ o $, the simulation of the step on $ p $ is ruled out, but the simulation of the step on $ q $ is still possible. Since either the simulation of the step on $ o $ or the simulation of the step on $ q $ succeeds, the simulation of the step on $ p $ is not only blocked but ruled out. But the consumption of the instantiation of the left lock does not unambiguously decide between the remaining two simulations. The intermediate state that results from consuming the instantiation of the left lock and represents a partial commitment is visualised in the right graph above by the state $ PC_1 $.

  Partial commitments forbid a 1-to-1 mapping between the states of  a source term and its translations by a bisimulation. But, as shown in \cite{parrowCoupled92}, partial commitments do not forbid to relate source terms and their translations by coupled similarity.

Whether the consumption of a positive instantiation of a lock is an auxiliary step---does not change the state of the term modulo $ \barbBisim $---, is a partial commitment, or unambiguously marks a simulation of a single source term step depends on the surrounding term, \ie cannot be determined without the context. For simplicity we consider all steps that reduce a positive instantiation of a lock as simulation steps.
Also steps on variants of the channels $ \much $, $ \rep $, and $ \renam'(X) $ are simulation steps, because they unambiguously mark a simulation of a single source term step. All remaining steps of the de-centralised encoding are auxiliary.

\begin{definition}[Auxiliary and Simulation Steps]
	A step $ T \stepT T' $ such that $ \exists S \in \procS\logdot \EncDO{S} \stepsT T $ is called a \emph{simulation step}, denoted by $ T \sStepT T' $, if $ T \step T' $ reduces a positive instantiation of a lock or is a step on a variant of $ \much $, $ \rep $, or $ \renam'(X) $.
	
	Else the step $ T \stepT T' $ is called an \emph{auxiliary step}, denoted by $ T \aStepT T' $.
	\label{def:auxStepsDecentral}
\end{definition}

\noindent
Again let $ \aStepsT $ denote the reflexive and transitive closure of $ \aStepT $ and let $ \sStepsT \deff \aStepsT \sStepT \aStepsT $.
Since auxiliary steps do not introduce partial commitments, they do not change the state modulo $ \barbBisim $. The proof of this lemma is very similar to the central case.

\begin{lemma}
	$ T \aStepT T' $ implies $ T \barbBisim T' $ for all target terms $ T, T' $.
	\label{lem:auxStepsDecentral}
\end{lemma}

\begin{proof}
	We distinguish the following cases \wrt the channel $ x $ that is reduced in the step $ T \aStepT T' $.
	\begin{compactenum}

		\item $ \boolV $ is a placeholder for $ t $ and $ f $, but, in contrast to $ t $ and $ f $, $ \boolV $ itself is never used as a channel name. Also $ \tau, \ch, \lch, \rch, z $, and $ \Proj{\Renam{\ch'}}{1} $ for all source term names $ \ch' $ are never used as channels.
		\item All variants of one of the names $ \act, \act', \lock', \req', \lreq', \rreq', \nextSyn, \syn $, and $ \syn' $ are used as simple forwarders. If we analyse the encoding functions in Figure~\ref{fig:innerEncoding} and Figure~\ref{fig:decentralised}, we observe that they are always restricted and there is exactly one replicated input and no other input on the respective variant in their scope. Thus, for all target terms $ T $ such that $ \exists S\logdot \EncCO{S} \steps T $, all steps on such channels satisfy the conditions specified by Proposition~\ref{prop:auxStepsForward}. Hence $ T \barbBisim T' $.
		\item The name $ \Proj{\Renam{\ch}}{3} $ is transmitted over $ \nextSyn $ in $ \Synch{\cdot} $ as initial value of $ \syn $. Thus, similarly to $ \syn $ because of Proposition~\ref{prop:auxStepsForward}, $ T \barbBisim T' $.
		\item The case of $ x $ being a variant of $ \req, \lreq, \rreq $ follows from Proposition~\ref{prop:auxStepsReq} and Proposition~\ref{prop:reqLocksDecentral}.
		\item The case of $ x $ being a variant of $ \Proj{\Renam{\ch}}{2} $ follows from 2.\@ and Proposition~\ref{prop:orderLeftAnnDecentral}.
		\item Variants of the names $ t, f $ are used to implement Boolean valued locks and an $ \ITE{\cdot}{\cdot}{\cdot} $-construct testing such locks. By Proposition~\ref{prop:sumLocksDecentral}, there is at most one positive instantiation of each lock and by definition all negative instantiations of the same lock---and also positive ones---are structural congruent. Since each $ \ITE{\cdot}{\cdot}{\cdot} $-construct restricts its own variants of $ t $ and $ f $ and because there is never a positive and a negative instantiation of the same lock (Proposition~\ref{prop:sumLocksDecentral}), all conflicts between two steps on variants of $ t $ and $ f $ result into structural congruent continuations and a step on variants of $ t $ and $ f $ cannot be in conflict with any other step on a different channel of $ T $ or its derivatives. Because $ \equiv \; \subseteq \; \barbBisim $ and by Lemma~\ref{lem:noConflicts}, then $ T \barbBisim T' $.
		\item Variants of the names $ \lock, \lLock, \rLock $ refer to Boolean valued locks. Again all announcements are propagated upwards---and on their way upwards some of them are composed---until they reach the outer layer $ \EncDO{\cdot} $. The de-central coordinator can process several announcements concurrently. Because of that conflicts result from different attempts to consume the same positive instantiation of a lock. However auxiliary steps can only consume negative instantiations of locks.
			By definition, all negative instantiations of the same lock are structural congruent. Moreover the encoding ensures that, as soon as the first negative instantiation of a lock is unguarded, as many negative instantiations of this lock are available as there are requests of it. The test of a negatively instantiated lock---that consumes an instantiation and reduces an $ \ITE{\cdot}{\cdot}{\cdot} $-construct---always reduces to the $ \ELSE{\cdot} $-case such that the inner part of a nested $ \ITE{\cdot}{\cdot}{\cdot} $-construct is not unguarded. Any unguarding of a $ \ITE{\cdot}{\cdot}{\cdot} $-construct also releases a request on the tested lock.
			Thus, if there are two negative instantiations for the same lock, it does not matter (modulo structural congruence) which one is reduced by a $ \ITE{\cdot}{\cdot}{\cdot} $-construct. Similarly, if there are two $ \ITE{\cdot}{\cdot}{\cdot} $-construct testing the same lock, both can be processed concurrently and it does not matter (modulo structural congruence) which consumes which instantiation.
			All steps on other channels cannot be in conflict with a step reducing a negative instantiation of the respective lock. Because $ \equiv \; \subseteq \; \barbBisim $ and by Lemma~\ref{lem:noConflicts}, then $ T \barbBisim T' $.
	\end{compactenum}
\end{proof}

In contrast to the centralised encoding, the simulation of a source term step in the de-centralised encoding can require more than a single simulation step and a single simulation step not unambiguously refers to the simulation of a particular source term step. The partial commitments described above forbid for operational correspondence, but the weaker variant proposed in \cite{gorla10} is satisfied. We call this variant weak operational correspondence.

\begin{definition}[Weak Operational Correspondence]
	$ $\\
	An encoding $ \enc: \procS \to \procT $ is \emph{weakly operationally corresponding} \wrt $ \barbCS \; \subseteq \procT^2 $ if it is:
	\begin{compactitem}
		\item[\; Complete:] $ \forall S, S' \logdot S \stepsS S' $ implies $ \exists T \logdot \EncDO{S} \stepsT T \wedge \EncDO{S'} \barbCS T $
		\item[\; Weakly Sound:] $ \forall S, T \logdot \EncDO{S} \stepsT T $ implies $ \exists S', T' \logdot S \stepsS S' \wedge T \stepsT T' \wedge \EncDO{S'} \barbCS T' $
	\end{compactitem}
\end{definition}

The only difference to operational correspondence is the weaker variant of soundness that allows for $ T $ to be an intermediate state that does not need to be related to a source term directly. Instead there has to be a way from $ T $ to some $ T' $ such that $ T' $ is related to a source term.

\begin{theorem}
	The encoding $ \encDO $ is weakly operational corresponding \wrt to $ \barbBisim $.
	\label{thm:operationalCorrespondenceDecentral}
\end{theorem}

\begin{proof}
	\begin{compactitem}
		\item[Completeness:] $ \forall S, S' \logdot S \stepsS S' $ implies $ \exists T \logdot \EncDO{S} \stepsT T \wedge \EncDO{S'} \barbBisim T $.\\
			We consider a single step $ S \stepsS S' $. Completeness then follows by induction on the number of steps in $ S \stepsS S' $.
			
			Assume $ S \stepS S' $. Since $ \encCO $ and $ \encDO $ differ only by the use of $ \once $, the simulation of source term steps is similar except for the one step on channel $ \once $. Hence the existence of $ T $ such that $ \EncDO{S} \stepsT T $ and $ \EncDO{S'} \barbBisim T $ can be proved by adapting the `if'-part of Lemma~\ref{lem:sourceVsSimStep} \wrt the step on $ \once $.
		\item[Weak Soundness:] $ \forall S, T \logdot \EncDO{S} \stepsT T $ implies $ \exists S', T' \logdot S \stepsS S' \wedge T \stepsT T' \wedge \EncDO{S'} \barbBisim T' $.\\
			By Lemma~\ref{lem:auxStepsDecentral}, it suffices to concentrate on the simulation steps in the sequence $ \EncDO{S} \stepsT T $. The proof is by induction on the number of simulation steps in the sequence $ \EncDO{S} \stepsT T $.
			
			In the \emph{base case}---without any simulation steps in $ \EncDO{S} \stepsT T $---choose $ S' = S $ and $ T' = T $ then $ S \stepsS S' $, $ T \stepsT T' $, and, by Lemma~\ref{lem:auxStepsDecentral}, $ \EncDO{S} = \EncDO{S'} \barbBisim T' = T $.
			
			Assume that there are $ S_H, T_H $ such that $ S \stepsS S_H $, $ T \stepsT T_H $, $ \EncDO{S_H} \barbBisim T_H $, and $ T \stepsT T_H $ contains only simulation steps necessary to resolve partial commitments (\emph{induction hypothesis}).
			
			Consider $ \EncDO{S} \stepsT T \sStepT T'' $.
			\begin{compactenum}
				\item A simulation step $ T \sStepT T'' $ that consumes a positive lock can result in partial commitment, but only in the case the respective reduced $ \ITE{\cdot}{\cdot}{\cdot} $-construct was the first part of a nested $ \ITE{\cdot}{\cdot}{\cdot} $-construct and the second part tests a lock $ \lock_2 $ of that a positive instantiation is (modulo auxiliary steps) still available. Switching the positive instantiation of $ \lock_2 $ into a negative instantiation---regardless of which $ \ITE{\cdot}{\cdot}{\cdot} $-construct is used to do so---resolves the partial commitment. The sequence $ \EncDO{S} \stepsT T $ might already introduce several more of such partial commitments. Proposition~\ref{prop:sumLocksDecentral} ensures some important properties over the instantiation of locks but it does not ensure, that for all locks there will eventually be an instantiation available. Only $ \ITE{\cdot}{\cdot}{\cdot} $-construct consume instantiations of locks. After being reduced, they restore all positive instantiation they consumed or turn them into negative instantiations. Negative instantiations remain available. Thus, to ensure that there are no deadlocks and all partial commitments can be resolved, we have to show that $ \ITE{\cdot}{\cdot}{\cdot} $-constructs cannot completely block each other. In the case of the centralised encoding this follows from the use of $ \once $. In the de-central encoding we make use of the same technique already used in \cite{peters12, petersNestmann15} to avoid this problem. As proved in \cite{peters12}, because we always consume first the instantiation of the lock from the left the nested $ \ITE{\cdot}{\cdot}{\cdot} $-constructs cannot all be blocked and we can resolve them step by step.
					
					In the present case the step $ T \sStepT T'' $ might consume an instantiation of a lock that was necessary for the sequence $ T \stepsT T_H $. If that is not the case, no step of $ T \stepsT T_H $ is in conflict with $ T \sStepT T'' $. Because of $ \EncDO{S'} \barbBisim T_H $, $ T_H $ does not contain unresolved partial commitments. Hence we can choose $ T' $ as the result of performing all steps of $ T \stepsT T_H $ in $ T'' $ followed, if necessary, by a sequence with a single simulation step $ \sStepsT $ to resolve the partial commitment that may result from $ T \sStepT T'' $ such that $ T \sStepT T'' \stepsT T' $. Then choose $ S' = S_H $, if no additional step was necessary to obtain $ T' $, else $ T \sStepT T'' $ and the additional simulation step are all simulation steps of the simulation of a source term step reducing action-prefixes and we choose $ S' $ as the result of performing the respective source term step in $ S_H $. Thus $ S \stepsS S' $. Because of $ \EncDO{S_H} \barbBisim T_H $ and the construction of $ S' $ and $ T' $, we have $ \EncDO{S'} \barbBisim T' $.
					
					Else, if there is a conflict between $ T \sStepT T'' $ and a step of $ T \stepsT T_H $, choose $ T' $ as the result of applying all but the conflicting step (and all auxiliary steps that depend on this step) of $ T \stepsT T_H $ in $ T'' $ followed, if necessary, by a sequence with a single simulation step $ \sStepsT $ to resolve the partial commitment that may result from $ T \sStepT T'' $. Because the induction hypotheses ensures that $ T \stepsT T_H $ contains only simulation steps necessary to resolve partial commitments, there are no simulation steps that depend on the conflicting step and all other simulation steps of $ T \stepsT T_H $ can be transferred to $ T'' $. Thus $ T \sStepT T'' \stepsT T' $. As a consequence of this replacement the simulation of a single source term step is replaced by the simulation of another single source term step. Choose $ S' $ as the result of replacing in $ S \stepsS S_H $ the respective source term step such that $ S \steps S' $. Because of $ \EncDO{S_H} \barbBisim T_H $, the construction of $ S' $ and $ T' $, and Observation~\ref{obs:guardedSourceVsTarget}, we have $ \EncDO{S'} \barbBisim T' $.
				\item A simulation step $ T \sStepT T'' $ on a variant of $ \much $ unguards exactly one source term encoding---let us call it $ \EncCI{P} $---due to the translation of internal choice. This step ensures that the respective other encoded source term alternative---let us call it $ \EncCI{Q} $---of the internal choice can never be unguarded, \ie is modulo $ \barbBisim $ similar to $ \Null $. This is the only effect of the steps $ T \sStepT T'' $ that can be observed modulo $ \barbBisim $.
					Since the encoding restricts each variant of $ \much $, the only step that can be in conflict with this step is the step unguarding $ \EncCI{Q} $. Because steps on a variant of $ \much $ do not resolve partial commitments, $ T \stepsT T_H $ does not contain such a step. Then, because $ T $ can perform the step, also $ T_H $ contains (modulo structural congruence) the unguarded subterm $ \Res{\much}{\Input{\much}{}{\EncCI{P}} \mid \Input{\much}{}{\EncCI{Q}} \mid \Out{\much}{}} $. By Observation~\ref{obs:guardedSourceVsTarget} and because $ \EncDO{S_H} \barbBisim T_H $, then $ \inchoice{P}{Q} $ is unguarded in $ S_H $.
					Hence we can choose $ S' $ by replacing $ \inchoice{P}{Q} $ in $ S_H $ by $ P $ such that $ S \stepsS S_H \stepS S' $ and we can choose $ T' $ by replacing $ \Res{\much}{\Input{\much}{}{\EncCI{P}} \mid \Input{\much}{}{\EncCI{Q}} \mid \Out{\much}{}} $ in $ T_H $ by $ \EncCI{P} $ such that $ T \sStepT T'' \stepsT T' $.
					By Observation~\ref{obs:guardedSourceVsTarget}, because of $ \EncDO{S_H} \barbBisim T_H $, and by the construction of $ S' $ and $ T' $, then $ \EncDO{S'} \barbBisim T' $.
				\item A simulation step $ T \sStepT T'' $ on a variant of $ \rep $ is due to the translation of $ \Div $. In this case $ T \barbBisim T'' $. Then, because of $ T \stepsT T_H $, there exists $ T' $ such that $ T'' \stepsT T' $, $ T'' \stepsT T' $ has the same simulation steps then $ T \stepsT T_H $, and $ T_H \barbBisim T' $. Thus we can choose $ S' = S_H $.
				\item A simulation step $ T \sStepT T'' $ on a variant of $ \renam'(X) $ is due to the translation of recursion. In this case $ T \sStepT T'' $ unguards exactly one source term encoding---let us call it $ \EncCI{P} $. This is the only effect of the steps $ T \sStepT T'' $ that can be observed modulo $ \barbBisim $.
					Since the encoding restricts each variant of $ \renam'(X) $ and the only input on this channel is replicated, this step is not in conflict with any other step of $ T $ or its derivatives. Because steps on a variant of $ \renam'(X) $ do not resolve partial commitments, $ T \stepsT T_H $ does not contain such a step. Then, because $ T $ can perform the step, also $ T_H $ contains (modulo structural congruence) the unguarded subterm $ \Res{\renam'(X)}{\Out{\renam'(X)}{} \mid \RepInput{\renam'(X)}{}{\EncCI{P}}} $. By Observation~\ref{obs:guardedSourceVsTarget} and because $ \EncDO{S_H} \barbBisim T_H $, then $ \mu X \cdot P $ is unguarded in $ S_H $.
					Hence we can choose $ S' $ as the result of replacing $ \mu X \cdot P $ in $ S_H $ by $ P\left[ \left( \mu X \cdot P \right) / X \right] $ such that $ S \stepsS S_H \stepS S' $ and we can choose $ T' $ as the result of replacing $ \Res{\renam'(X)}{\Out{\renam'(X)}{} \mid \RepInput{\renam'(X)}{}{\EncCI{P}}} $ in $ T_H $ by $ \Res{\renam'(X)}{\Out{\renam'(X)}{} \mid \RepInput{\renam'(X)}{}{\EncCI{P}} \mid \EncCI{P}} $, where all occurrences of $ X $ in $ P $ are translated to $ \Out{\renam'(X)}{} $, such that $ T \sStepT T'' \stepsT T' $.
					By Observation~\ref{obs:guardedSourceVsTarget}, because of $ \EncDO{S_H} \barbBisim T_H $, and by the construction of $ S' $ and $ T' $, then $ \EncDO{S'} \barbBisim T' $.
			\end{compactenum}
	\end{compactitem}
\end{proof}

As in the encoding $ \encCO $, there is no infinite sequence of only auxiliary steps in $ \EncDO{S} $.

\begin{lemma}
	The number of steps between two simulation steps is finite.
	\label{lem:numberStepsDecentral}
\end{lemma}

\begin{proof}
	In contrast to $ \encCO $, also the consumption of announcements by the coordinator is a simulation step for $ \encDO $. Since there are only finitely many announcements, consuming them does not lead to divergence. Because of the consumption of announcements, the coordinator can release several requests $ \Out{\req}{} $, but again only finitely many. Accordingly, the sequences of auxiliary steps can be longer in $ \encDO $, but all such sequences result from interleaving finitely many sequences of auxiliary steps of $ \encCO $.
	Apart from these observations the proof is similar to the proof of Lemma\ref{lem:numberStepsCentral}.
\end{proof}

Moreover each simulation of a source term requires only finitely many simulation steps (to consume the respective positive instantiations of locks). Thus $ \encDO $ reflects divergence.

\begin{theorem}
	The encoding $ \encDO $ reflects divergence.
	\label{thm:divergenceReflectionDecentral}
\end{theorem}

\begin{proof}
	If $ \EncDO{S} $ is divergent then, by Lemma~\ref{lem:numberStepsDecentral}, $ \EncDO{S} $ can perform an infinite sequence of steps containing infinitely many simulation steps.
	Simulation steps either directly represent the simulation of a source term step---as in the case of recursion, divergence, and internal choice---or reduce a positive instantiation of a lock.
	Instantiations of locks are consumed by $ \ITE{\cdot}{\cdot}{\cdot} $-constructs. These constructs are guarded by requests $ \Out{\req}{} $.
	By Lemma~\ref{lem:numberStepsDecentral}, without simulation steps only finitely many requests and thus instantiations of locks can be consumed. Simulation steps can only lead to new requests if they unguard the translation of a source term continuation, but then the simulation of a source term step was completed.
	Hence, if $ \EncDO{S} $ is divergent, then $ S $ is divergent.
\end{proof}

The encoding function ensures that $ \EncDO{S} $ has an unguarded occurrence of $ \success $ iff $ S $ has such an unguarded occurrence. Operational correspondence again ensures that $ S $ and $ \EncDO{S} $ also answer the question for the reachability of $ \success $ in the same way.

\begin{theorem}
	The encoding $ \encDO $ is success sensitive.
	\label{thm:successSensitivenessDecentral}
\end{theorem}

\begin{proof}
	From Observation~\ref{obs:success} and Figure~\ref{fig:decentralised}, $ S\hasSuccess $ iff $ \EncDO{S}\hasSuccess $. With Theorem~\ref{thm:operationalCorrespondenceDecentral} and because $ \barbBisim $ respects $ \success $, then $ S\reachSuccess $ iff $ \EncDO{S}\reachSuccess $.
\end{proof}

Similarly, a source term reaches a barb iff its translation reaches the respective translated barb.

\begin{theorem}
	$ \forall S, \ch\logdot S\ReachBarb{\ch} $ iff $ \EncDO{S}\ReachBarb{\EncCI{\cdot}\ch} $.
	\label{thm:respectsBarbsDecentral}
\end{theorem}

\begin{proof}
	From Observation~\ref{obs:transBarbs} and Figure~\ref{fig:decentralised}, $ S\HasBarb{\ch} $ iff $ \EncDO{S}\HasBarb{\EncCI{\cdot}\ch} $. With Theorem~\ref{thm:operationalCorrespondenceDecentral} and because $ \barbBisim $ respects translated barbs, then $ S\ReachBarb{\ch} $ iff $ \EncDO{S}\ReachBarb{\EncCI{\cdot}\ch} $.
\end{proof}

Weak operational correspondence does not suffice to establish a bisimulation between source terms and their translations.
But, as proved in \cite{petersGlabbeek15}, Theorem~\ref{thm:operationalCorrespondenceDecentral}, the fact that $ \barbBisim $ is success sensitive and respects (translated) observables, Theorem~\ref{thm:successSensitivenessDecentral}, and Theorem~\ref{thm:respectsBarbsDecentral} imply that $ \forall S\logdot S $ and $ \EncCO{S} $ are (success sensitive, (translated) barbs respecting, weak, reduction) coupled similar, \ie $ S \barbCS \EncDO{S} $.

It remains to show, that $ \encDO $ indeed preserves distributability. Therefore we prove that all blocking parts of the encoding $ \encDO $ refer to simulations of conflicting source term steps.

\begin{theorem}
	The encoding $ \encDO $ preserves distributability.
	\label{thm:distributability}
\end{theorem}

\begin{proof}
	The de-central coordinator in Figure~\ref{fig:decentralised} computes announcements concurrently. The test of locks is technically an output and steps on $ t $ and $ f $ are restricted such that these steps are never in conflict to any other step. Thus the de-central coordinator itself does not block the concurrent simulation of distributable steps.
	
	In $ \encCI $ all blocking, \ie all not-replicated inputs, are on variants of $ \Proj{\Renam{\ch}}{2}, \req, \lock, \lLock, \rLock, \much $. Two steps on the same variant of $ \Proj{\Renam{\ch}}{2} $ belong to two simulation attempts of source term steps on the same action that needs to be synchronised by a parallel operator. Since such steps are also not distributable in the source, their simulations do not have to be distributable.
	
	Two steps on the same variant of one of the names $ \req, \lock, \lLock, \rLock $ belong to simulation attempts that need to consume the same positive instantiation of a lock. Thus these two attempts clearly try to simulate conflicting source term steps. Hence again the two simulation attempts do not have to be distributable.
	
	Similarly two steps on the same variant of $ \much $ clearly belong to two simulation attempts of conflicting source term steps. Thus again they do not have to be distributable.
	
	We conclude that the simulations of distributable source terms are distributable, \ie $ \encDO $ preserves distributability.
\end{proof}

\section{Conclusions}
\label{sec:conclusion}

We introduced two encodings from CSP into asynchronous CCS with name passing and matching.
As in \cite{parrowCoupled92} we had to encode the multiway synchronisation mechanism of CSP into binary communications and, similarly to \cite{parrowCoupled92}, we did so first using a central controller that was then modified into a de-central controller.
By doing so we were able to transfer the observations of \cite{parrowCoupled92} to the present case:
\begin{compactenum}
	\item The central solution allows to prove a stronger connection between source terms and their translations, namely by bisimilarity. Our de-central solution does not relate source terms and their translations that strongly and we doubt that any de-central solution can do so.
	\item Nonetheless, de-central solutions are possible as presented by the second encoding and they still relate source terms and their translations in an interesting way, namely by coupled similarity.
\end{compactenum}
Thus as in \cite{parrowCoupled92} we observed a trade-off between \emph{central} but \emph{bisimilar} solutions on the one-hand side and \emph{coupled similar} but \emph{de-central} solutions on the other side.

More technically we showed here instead a trade-off between central but \emph{operational corresponding} solutions on the one-hand side and \emph{weakly operational corresponding} but de-central solutions on the other side.
The mutual connection between operational correspondence and bisimilarity as well as between weak operational correspondence and coupled similarity is proved in \cite{petersGlabbeek15}.

Both encodings make strict use of the renaming policy and translate into closed terms.
Hence the criterion \emph{name invariance} is trivially satisfied in both cases.
Moreover we showed that both encodings are \emph{success sensitive}, \emph{reflect divergence}, and even \emph{respect barbs} \wrt to the standard source term (CSP) barbs and a notion of translated barbs on the target.
The centralised encoding $ \encCO $ additionally satisfies a variant of \emph{operational correspondence} that is stricter than the variant proposed in \cite{gorla10}.
The de-centralised encoding $ \encDO $ satisfies \emph{weak operational correspondence} as proposed in \cite{gorla10} and \emph{distributability preservation} as proposed in \cite{petersNestmannGoltz13}.
Thus both encodings satisfy all of the criteria proposed in \cite{gorla10} except for compositionality.
However in both cases the inner part is obviously compositional and the outer part adds only a fixed context.

\end{document}